\newif\ifdraft\draftfalse%
\newif\ifanon\anonfalse%
\newif\iffull\fullfalse%

\newif\iflongrefs\longrefsfalse%
\newif\ifbackref\backreffalse%
\newif\ifsooner\soonerfalse%
\newif\iflater\laterfalse%
\newif\ifcamera\camerafalse%
\newif\ifneedspace\needspacefalse%
  \documentclass[10pt, conference, compsocconf, times]{IEEEtrancustom}
\IEEEoverridecommandlockouts%

\renewcommand{\ttdefault}{cmtt}

\DeclareMathAlphabet{\mathit}{\encodingdefault}{\familydefault}{m}{it}

\usepackage[numbers,sort]{natbib}

\usepackage{amsmath,amssymb,amsfonts,amsthm}
\usepackage{algorithmic}
\usepackage{graphicx}
\usepackage{textcomp}
\usepackage[dvipsnames]{xcolor}
\usepackage{enumitem}
\usepackage{combelow}

\usepackage[frozencache]{minted}

\definecolor{dkblue}{rgb}{0,0.1,0.5}
\definecolor{dkgreen}{rgb}{0,0.3,0}
\definecolor{dkred}{rgb}{0.6,0,0}
\definecolor{dkpurple}{rgb}{0.7,0,1.0}
\definecolor{purple}{rgb}{0.9,0,1.0}
\definecolor{olive}{rgb}{0.4, 0.4, 0.0}
\definecolor{teal}{rgb}{0.0,0.4,0.4}
\definecolor{azure}{rgb}{0.0, 0.5, 1.0}
\definecolor{gray}{rgb}{0.5, 0.5, 0.5}
\definecolor{dkgray}{rgb}{0.3, 0.3, 0.3}

\ExplSyntaxOn
\NewDocumentCommand{\code}{m}{
  \tl_set:Nn \l_tmpa_tl {#1}
  \tl_replace_all:Nnn \l_tmpa_tl {'} {\textquotesingle}
  \texttt{\l_tmpa_tl}
}
\ExplSyntaxOff
\newenvironment{mprules}{\[\hbox\bgroup\footnotesize\def\mpsep{&}$\begin{array}{@{}r@{\;}c@{\;}l@{}}}{\end{array}$\egroup\]}
\newcommand{\mpsep}{\;}
\newcommand{\mprule}[3]{\left[\begin{array}{@{}l@{}}#2\end{array}\right] \mpsep \mathtt{#1} \mpsep \left[\begin{array}{@{}l@{}}#3\end{array}\right]}
\newcommand{\highlight}[1]{{\setlength{\fboxsep}{1pt}\colorbox{Goldenrod!35}{\(#1\)}}}

\usepackage{tikz}
\usetikzlibrary{positioning, arrows.meta, bending, calc, decorations.pathreplacing}

\usepackage{subcaption}
\usepackage{mymacros}
\usepackage{mycmds}

\theoremstyle{plain}
\newtheorem{theorem}{Theorem}
\newtheorem{lemma}{Lemma}

\theoremstyle{definition}
\newtheorem{definition}{Definition}
\newtheorem{example}{Example}
\newtheorem{proposition}{Proposition}

\usepackage[
    pdftex,%
    pdfpagelabels,%
    colorlinks,%
    linkcolor=dkblue,%
    citecolor=dkblue,%
    filecolor=dkblue,%
    urlcolor=dkblue%
    \ifbackref,backref\fi%
]{hyperref}
\usepackage{cleveref}

\definecolor{ann0}{HTML}{6E6E6E}
\definecolor{ann1}{HTML}{1B7A3D}
\definecolor{ann2}{HTML}{A03A00}
\definecolor{ann3}{HTML}{1B4F9C}
\definecolor{ann4}{HTML}{6B2D8E}
\definecolor{ann5}{HTML}{A8730A}
\definecolor{ann6}{HTML}{0E7C86}
\definecolor{ann7}{HTML}{A32A5B}
\makeatletter
\NewDocumentCommand{\anndecl}{O{} m m m}{%
  \edef\ann@hue{#1}\ifx\ann@hue\@empty\edef\ann@hue{#3}\fi
  \Hy@raisedlink{\hyper@anchorstart{ann.#2}\hyper@anchorend}%
  \protected@write\@auxout{}{\string\newlabel{#2}{{#3#4}{\thepage}{#3#4}{ann.#2}{}}}%
  \protected@write\@auxout{}{\string\newlabel{#2@cref}%
    {{[ann][#3#4][]\string\annrefstyle{\ann@hue}{#3#4}}{[1][\thepage][]\thepage}{}{}{}}}%
}
\makeatother
\newcommand{\annrefstyle}[2]{{\color{ann#1}\textup{(#2)}}}
\crefformat{ann}{#2#1#3}
\Crefformat{ann}{#2#1#3}
\Crefformat{section}{#2\S#1#3}
\Crefformat{subsection}{#2\S#1#3}
\Crefmultiformat{section}{#2\S#1#3}{ and~#2\S#1#3}{, #2\S#1#3}{ and~#2\S#1#3}
\newcommand{\annbadge}[2]{\text{{\color{ann#1}\scriptsize\textup{(#2)}}}}
\newcommand{\annbox}[3]{\tikz[baseline=(n.base)]{%
    \node[draw=ann#1,line width=.4pt,rounded corners=1.5pt,
          inner xsep=2.5pt,inner ysep=2pt](n){#3};}\,\annbadge{#1}{#2}}

\newcommand{\titleString}{Blame-aware Recomposition for Secure Compilation Proofs of Low-Level Compartment Isolation}

\title{\titleString}
\makeatletter
\AtBeginDocument{
  \hypersetup{
    pdftitle = {\titleString},
    pdfauthor = { Schlosser, Thibault, and Hrițcu}
  }
}
\makeatother
\author{
  Basile Schlosser\textsuperscript{1,3} \qquad 
  Jérémy Thibault\textsuperscript{2,3} \qquad %
  C\u{a}t\u{a}lin Hri\cb{t}cu\textsuperscript{3} \quad %
\\[0.5em]
{\small
  \textsuperscript{1}ENS Paris-Saclay, France\qquad
  \textsuperscript{2}EPFL, Lausanne, Switzerland\qquad
  \textsuperscript{3}MPI-SP, Bochum, Germany
}
}

\begin{document}

\maketitle

\begin{abstract}
  The security guarantees of existing formally verified compilers for
compartmentalized programs in unsafe languages like C stop at a
rather high level, assuming a built-in protection mechanism and a block-based
memory model. Verifying a concrete, low-level enforcement mechanism for
compartment isolation is still a major challenge for current proof techniques.
In particular, Abate et al.'s proof technique fails because one of their proof
steps, recomposition, essentially requires hiding all information
about the code of each compartment, which is not practically achievable at a low level.
We propose an extension of their proof technique to a low-level setting where
compartments can fully observe each other's code. To
do so, we introduce an overlay semantics that allows us to characterize these
observations as a form of undefined behaviour, and we use a novel variant of
recomposition that we call blame-aware recomposition, which essentially
allows recomposing executions as needed
in the original high-level proof technique, even if the low-level context observes the
code, because this context is ultimately going to be discarded.
We apply this new proof to a compartment isolation mechanism
that uses a programmable tag-based reference monitor.

\end{abstract}

\section{Introduction}

Vulnerabilities in software systems can lead to devastating attacks when
exploited by malicious attackers.
In particular, unsafe languages such as C call out many \emph{undefined
behaviours} (UB), which include various memory safety vulnerabilities that account
for 70\% of high-severity security bugs, according to both Microsoft and Chrome~\cite{memory-unsafety-chrome,memory-unsafety-microsoft-zdnet}.

Compartmentalization~\cite{YeeSDCMOONF10,ReisG09,sfi_sosp1993,Kilpatrick03,GudkaWACDLMNR15} is an often proposed practical countermeasure
against vulnerabilities that relies on the principle of least privilege:
programs are split into mutually distrustful compartments that interact
through protected interfaces.
If one compartment is compromised, the security impact on the rest
of the program is limited.
Compartmentalization is widely deployed in web browsers,
which use process-level isolation to protect web sites from each
other~\cite{ReisG09,ReisMO19}
and software-fault isolation (SFI)~\cite{sfi_sosp1993,YeeSDCMOONF10,Tan17}
to sandbox WebAssembly code and third-party
libraries~\cite{NarayanDGFRLSS20,KolosickNJWLGJS22}.

In order to characterize precisely the security guarantees a compilation chain for
compartmentalized source programs should satisfy in the presence of undefined
behaviour, Abate \etal~\cite{AbateABEFHLPST18,Abate_2018_arxiv} propose a criterion:
secure compilation should
restrict undefined behaviours both spatially (to just the compartments that
encounter undefined behaviour), and temporally (compartments are protected until
they encounter an undefined behaviour).
They reduce their criterion to a variant of a secure compilation criterion known
as \emph{Robust Safety Preservation} (RSP)~\cite{AbateBGHPT19,PatrignaniG21} that
they call \(\rscdcmd\), standing for \emph{Robustly Safe Compilation with
  Dynamic Compromise and Mutual Distrust}.
Their variant states that any attack mounted by low-level attackers (\ie,
low-level compartments that may have been dynamically compromised by undefined
behaviour) against a safety property of a compiled program can also be mounted by
source attackers (\ie, source-level compartments) against the source program.
In addition, they formally verify that a compiler for a simple compartmentalized
source language is secure according to this criterion.
Since then, this result has been extended to incorporate memory
sharing~\cite{secureptrs}, to prove stronger secure compilation
criteria~\cite{ThibaultLH25}, and most prominently to build
SECOMP~\cite{ThibaultBLAAGHT24}, a formally secure compiler based on CompCert for a
C variant with compartments.

However, a common limitation of these works is that they do not verify secure
compilation all the way down to the actual machine; instead, they stop at a
relatively high-level assembly language with an abstract view of memory and with
an abstract model of compartments built into the semantics.
Instead of a realistic flat memory model, where code is laid out in memory, they
assume a block-based memory model~\cite{LeroyB08} which only allows for a limited amount of memory
unsafety (for instance, they cannot model heap overflows), and have code in a
separate part of the state (so it cannot be modified or even read during
execution).
Hence, their abstract model of compartments is just a specification of what an
actual compartment isolation mechanism should do at a lower
level.

Abate \etal~\cite{AbateABEFHLPST18} implement this specification in two
lower-level compiler backends: one based on SFI~\cite{sfi_sosp1993,ReisG09,Tan17},
and another one based on a compartment isolation \emph{micro-policy}~\cite{micropolicies2015}, a programmable
tag-based reference monitor. %
They leave verifying these backends open, and instead systematically test
their SFI backend.
Similarly, SECOMP~\cite{ThibaultBLAAGHT24} proposes an unverified lower-level
backend targeting a variant of the CHERI RISC-V capability machine~\cite{CHERI-ISA}.

These works cannot verify such realistic backends because their proof
technique relies on a step called \emph{recomposition}~\cite{Abate_2018_arxiv}, which allows one to split two whole programs (obtained by linking partial programs together, denoted by a bowtie \(\bowtie\)) that produce the same event trace (calls, returns and outputs; written \(\semtrg m\)) into two parts each, and then recombine these parts into a new whole program that produces the same event trace.
In their proof, recomposition lets them replace a target-level
context $\trg{C_{\trglang}}$ with a compiled one $\cmp{\src{C_{\srclang}}}$:\footnote{We
  use \(\src{\textbf{bold, dark blue}}\) for source objects,
  and \(\trg{\text{thin, red}}\) for target objects.}
\[
  \forall\, m,
  \left.
  \begin{alignedat}{3}
    \trg{C_{\trglang}} \ & \linktrg \ & \cmp{\src{P}} & \semtrg m \\
    \cmp{\src{C_{\srclang}}} \ & \linktrg \ & \cmp{\src{P'}} & \semtrg m
  \end{alignedat}
  \right\}
  \implies
  \cmp{\src{C_{\srclang}}} \linktrg \cmp{\src{P}} \semtrg m
\]
This step relies on the event traces capturing all possible interactions
between compartments, meaning that the traces need to be
``fully abstract''~\cite{JeffreyR05,Laird07,PatrignaniC15}, and for this,
compartments cannot observe other compartments' \emph{code}.

This has little chance to hold at a low level, because there are many new observations compartments can make that are invisible (i.e. they don't appear on event traces) and hard to prevent.
A simple example would be reading code directly.
\begin{example}[Reading code directly]\label{ex:direct-read}
  Consider the following function \code{f} belonging to a compartment \code{C1}:
\begin{minted}[fontsize=\small]{c}
C1.f() {
  x := 100; // guesses an address
  if (*x < 1000) { // reads at it
    ...
  } else { ... }
}
\end{minted}
This function exploits the fact that there is no difference between a pointer and an integer in the low-level language to read memory at an arbitrary location, and branch on the result.
That is a problem, because after code is laid out in memory, address \code{100} may point to the code of another compartment, so the branch taken by \code{f} depends on the code of the compartments \code{C1} is linked with!
Hence, compartment \code{C1} learns information about other compartments' code, which is not recorded on the trace.
\end{example}

A second example shows observations through allocation.
\begin{example}[Observation through allocation]\label{ex:alloc-observation}~
\begin{minted}[fontsize=\small]{c}
C1.f() {
  x := malloc(1); // allocates one byte
  C2.g();
  y := malloc(1); // allocates another byte
  if (y - x > 1) { // compares the two addresses
    ...
  } else { ... }
}
\end{minted}
This function allocates a byte, calls a function from another compartment,
\code{C2.g}, and then allocates another byte.
By comparing the addresses of the two allocations, \code{C1} obtains information about the internal behaviour of another compartment, and can branch on it; this breaks recomposition, because allocating and comparing the two addresses are not events that are recorded on the event trace.
\end{example}

With a high-level memory model such as the block-based one of CompCert~\cite{LeroyB08}, these observations are not possible, even in CompCert's Asm: subtracting or ordering pointers to different blocks is undefined behaviour, and so is dereferencing an integer.
This is why previous recomposition proofs~\cite{Abate_2018_arxiv,ThibaultBLAAGHT24} did not run into them: executions stop with undefined behaviour before performing such observations.
Compiling further down to a flat memory model turns this undefined behaviour into defined low-level behaviour, so even a compartment compiled from source can perform such observations once it has been compromised by undefined behaviour, which is exactly the dynamic compromise scenario that \(\rscdcmd\) is designed to protect against.

We propose a significant extension of the recomposition-based proof technique of Abate \etal~\cite{Abate_2018_arxiv} that enables the verification of a secure compilation backend targeting a low-level assembly with a realistic memory model and a compartmentalization mechanism, which we call the low-level backend.
Our technique is based on two observations: first, that the partial programs kept in recomposition are both compiled, so they do not trigger undefined behaviours along the considered trace; and second, that one can characterize most behaviours that break recomposition at the lower level as undefined behaviours \wrt an overlay semantics~\cite{SkorstengaardDB19,GeorgesTB22}.

Specifically, we introduce an intermediate language that is all but equal to the low-level language targeted by the backend, except that it considers these behaviours undefined.
We write \(\ovsemtrg\) for the overlay semantics (whereas we use \(\semtrg\) for the low-level semantics).

We then replace the use of recomposition by a new variant we call \emph{blame-aware
recomposition}.
This variant allows recomposing partial programs just like before, with the
added precondition that three out of the four partial programs behave in a
defined way according to the overlay semantics, \ie, do not have the behaviours that
would break recomposition.
More specifically, blame-aware recomposition takes two executions:
\begin{align*}
  \text{if }\cmp{\src{C_S}} & \linktrg \cmp{\src{P'}} \ovsemtrg m &&\text{in the overlay semantics}\\
  \text{and either } &\trg{C_T} \linktrg \cmp{\src{P}} \ovsemtrg m &&\text{in the overlay semantics}\\
  \text{or } &\trg{C_T} \linktrg \cmp{\src{P}} \semtrg m &&\text{in the low-level semantics,}\\
  \textit{but} & &&\text{all undefined behaviours}\\
                 & &&\text{in this execution happen in } \trg{C_T}
\end{align*}
and yields an execution
of the recomposed program \(\cmp{\src{C_S}} \linktrg \cmp{\src{P}}\):
\begin{align*}
  \cmp{\src{C_S}} \linktrg \cmp{\src{P}} &\ovsemtrg m &&\text{in the overlay semantics.}
\end{align*}
The preconditions ensure that the observations that would break recomposition
can only occur in \(\trg{C_T}\), which is discarded, while
\(\cmp{\src{C_S}}\) cannot perform such an observation, at least along the
trace \(m\), because it executes in the overlay semantics.
This allows proving blame-aware recomposition for low-level
backends that don't explicitly prevent these observations.

Our blame-aware recomposition proof technique is strong enough to deal with observations that can be cleanly characterized in an overlay semantics, such as \Cref{ex:direct-read}.
The observations of \Cref{ex:alloc-observation} would be difficult to characterize as undefined behaviour in an overlay semantics without overly restricting acceptable behaviours (for instance, comparing pointers to the same block is allowed in CompCert-like memory models, and we would like to preserve this ability) or resorting to very complex tracking schemes.
Hence, instead of dealing with \Cref{ex:alloc-observation} using overlay semantics, we address it with a compiler-based masking scheme to hide compartments' heap usage: each compartment is statically assigned its own space in memory, ensuring that all of a compartment's memory allocations are contiguous independently of other compartments.
An advantage of this solution is that it prevents the allocation channel from leaking information not only about other compartments' code, but also about their data; the latter would be required to extend our guarantees beyond robust safety preservation in the future~\cite{AbateBGHPT19}.

We updated the formerly unverified micro-policy backend of Abate \etal~\cite{AbateABEFHLPST18} to close leaks of this kind, with an allocator that uses our masking scheme.
We prove that the resulting backend satisfies blame-aware recomposition using an overlay semantics.
Our overlay semantics is constructed as a variant of the micro-policy used at the lowest level: in addition to using the enforcement from the concrete backend, the overlay semantics calls reading and writing to code in memory and executing data undefined behaviours.

To prove blame-aware recomposition in this instance, we reuse the proof diagrams introduced to prove recomposition in SECOMP~\cite{ThibaultBLAAGHT24}.
While proving the main diagrams of blame-aware recomposition, we are, as expected, forced to reason about the concrete implementation of the compartment isolation mechanism, which is highly relevant for security.

In summary, we put forward the following contributions:
\begin{enumerate}
  \item We introduce \emph{blame-aware recomposition}, a novel variant of recomposition that relies on an intermediate \emph{overlay semantics} and gives rise to a proof technique that allows one to prove the security of a compiler targeting a low-level backend with a flat memory model.
  \item We generically prove that blame-aware recomposition can be used to replace recomposition and obtain \(\rscdcmd\). Our technique
  is formalized in the Rocq proof assistant.
  \item We apply our methodology to the micro-policy backend of Abate \etal~\cite{Abate_2018_arxiv}, modified as described above, proving blame-aware recomposition and as a consequence the security of this backend, in the Rocq proof assistant.
\end{enumerate}

\textbf{Outline:}
The rest of this paper is structured as follows: in \Cref{sec:new-approach}, we describe our new approach in a generic setting. In \Cref{sec:backend-change}, we discuss our adaptation of the micro-policy backend to make the proofs work. In \Cref{sec:recomp-proofs}, we discuss how we prove blame-aware recomposition for this backend. Finally, \Cref{sec:related,sec:future} discuss related and future work.

\textbf{Artifact:}
Our formal development is available at the following URL:
\url{https://github.com/secure-compilation/when-good-components-go-bad/tree/mp-transitional}

\section{General Machine-Checked Proof Technique Using Blame-Aware
  Recomposition}\label{sec:new-approach}

We start by presenting the proof technique of Abate \etal~\cite{Abate_2018_arxiv} in a generic setting, then explain why the same technique doesn't work fully in low-level settings, and finally describe a way to adapt this proof technique to such a setting.

\subsection{Setting and Background}
The generic setting we consider is as follows: we have an unsafe source language \(\srclang\) and a target language \(\trglang\).
Both languages have a notion of contexts (\(\src{C} \in \src{\mathbb{C}_{\mathcal{S}}}\) or \(\trg{C} \in \trg{\mathbb{C}_{\mathcal{T}}}\)) and partial programs (\(\src{P} \in \src{\mathbb{P}_{\mathcal{S}}}\) or \(\trg{P} \in \trg{\mathbb{P}_{\mathcal{T}}}\)) that can be linked together to obtain whole programs \(\src{W} = \src{C} \linksrc \src{P}\) or \(\trg{W} = \trg{C} \linktrg \trg{P}\).

Both languages are equipped with semantics that produce finite traces for whole programs: \(\src{W} \semsrc m\) and \(\trg{W} \semtrg m\).
The traces record stuckness with an event \(\bot_{\mathbb{P}}\) or \(\bot_{\mathbb{C}}\), depending on whether the stuckness occurred while executing in the partial program or the context.
Other than this, we do not require a particular structure for the traces, but a common choice is for them to record calls and returns between the partial program and the context.
We write \(\leq\) for the prefix relation between finite traces, and \(m' \prec_{\mathbb{P}} m\) (resp.\ \(m' \prec_{\mathbb{C}} m\)) to mean \(\exists\,m'' \leq m, m' = m'' \cdot \bot_{\mathbb{P}}\) (resp.\ \(\bot_{\mathbb{C}}\)): it captures \(m\) being a refinement of an undefined behaviour that appears at the end of \(m'\) while executing inside the partial program (resp.\ the context).

Additionally, our proof technique relies on small-step operational semantics at the target level to be able to define and to construct intermediate semantics that will be used in our replacement of recomposition.
Hence, we require the semantics of \(\trglang\), which we call \emph{low-level semantics}, to be described by a small-step relation between runtime states \(\trg{\xrightarrow{\com{e}}}\) where \(e\) can be either an event or nothing (for silent steps).

We write \(\trg{\xrightarrow{\com{m}}}^*\) to denote the transitive reflexive closure of that relation, producing the trace \(m\).
We suppose the initial state of a program \(\trg{W}\) is given by \(\operatorname{\trg{initial}}(\trg{W})\) and we see \(\trg{W} \semtrg m\) as a shorthand for \(\exists\, {\trg{s}}, \operatorname{\trg{initial}}(\trg{W}) \trg{\xrightarrow{\com{m}}}^* \trg{s}\).
Importantly, target-level runtime states should also be classified as belonging either to the program or the context, depending on which of the components is currently executing, written \(\trg{s} \in \trg{\mathbb{P}}\) or \(\trg{s} \in \trg{\mathbb{C}}\).

On the other hand, it doesn't matter how the semantics of \(\srclang\) is
defined, except for the fact that the back-translation and compiler correctness lemmas below have to be provable.

Our goal is to prove that a compiler \(\cmp[\srclang][\trglang]{}\) from \(\srclang\) to \(\trglang\) is secure, in the sense that it satisfies \(\rscdcmd\)~\cite{AbateABEFHLPST18}.
\begin{definition}[\(\rscdcmd\)]
  \begin{align*}
    &\forall\,\src{P}\in\src{\mathbb{P}_{\mathcal{S}}},\trg{C}\in\trg{\mathbb{C}_{\mathcal{T}}},m,\
     \trg{C} \linktrg \cmp[\srclang][\trglang]{\src{P}} \semtrg m \implies \\
    &\quad\exists\,\src{C}\in\src{\mathbb{C}_{\mathcal{S}}},m',\
     \src{C} \linksrc \src{P} \semsrc m' \wedge
      (m \leq m' \vee m' \prec_{\mathbb{P}} m).
  \end{align*}
\end{definition}
We refer the interested reader to~\cite{AbateABEFHLPST18} for a detailed explanation of what makes \(\rscdcmd\) interesting, but in short, \(\rscdcmd\) captures the fact that a low-level attack by a low-level context \(\trg{C}\), represented by a trace \(m\), can be explained in terms of a source context \(\src{C}\) that can perform an attack \(m'\) that is either longer than \(m\) (\(m \leq m'\)), or that can be blamed on (explained in terms of) an undefined behaviour triggered in the source program \(\src{P}\) (\(m' \prec_{\mathbb{P}} m\)).
The second disjunct rules out vacuous back-translations: the constructed context has defined behaviour, so any UB of \(\src{C} \linksrc \src{P}\) comes from \(\src{P}\).

To use the proof technique from Abate \etal, we need four components: \emph{compiler correctness}, a correct \emph{back-translation} from finite traces to source whole programs, the \emph{recomposition} theorem, and the \emph{blame} theorem.
Following the proof of Abate \etal, we start from the hypothesis \(\trg{C} \linktrg \cmp[\srclang][\trglang]{\src{P}} \semtrg m\) and apply the back-translation lemma:
\begin{proposition}[Back-translation]\label{hyp:back-translation}
  There exists \(\uparrow\), a function from finite traces to whole programs, such that:
  \begin{align*}
    &\forall\,\trg{P}\in\trg{\mathbb{P}_{\mathcal{T}}}, \trg{C}\in\trg{\mathbb{C}_{\mathcal{T}}},m,\
     \trg{C} \linktrg \trg{P} \semtrg m \implies\\
    &\quad\exists\,\src{C_S}\in\src{\mathbb{C}_{\mathcal{S}}},\src{P'}\in\src{\mathbb{P}_{\mathcal{S}}},\
     m\mathord{\uparrow} = (\src{C_S} \linksrc \src{P'}) \wedge \src{C_S} \linksrc \src{P'} \semsrc m.
  \end{align*}
\end{proposition}
The back-translation allows one to take a (well-formed) finite trace \(m\) and
obtain a whole source program that produces this exact trace.
With this, we obtain a new \(\src{C_{\mathcal{S}}}\) and a new \(\src{P'}\) in the source.
The next step is to compile these back to the target, and to apply the forward direction of compiler correctness, which states that a compiled program produces all traces that were produced by the source program.
\begin{proposition}[Compiler Correctness (forward direction)]\label{hyp:fcc}
  \begin{align*}
    &\forall\src{P}\in\src{\mathbb{P}_{\mathcal{S}}}, \src{C}\in\src{\mathbb{C}_{\mathcal{S}}},m,
     \src{C} \linksrc \src{P} \semsrc m \implies
     \cmp[\srclang][\trglang]{\src{C}} \linktrg \cmp[\srclang][\trglang]{\src{P}} \semtrg m.
  \end{align*}
\end{proposition}
After applying this theorem, we now have both \(\trg{C} \linktrg \cmp[\srclang][\trglang]{\src{P}} \semtrg m\) (from the hypothesis) and \(\cmp[\srclang][\trglang]{\src{C_{\srclang}}} \linktrg \cmp[\srclang][\trglang]{\src{P'}} \semtrg m\) (from compiler correctness).
This is where recomposition gets applied:
\begin{definition}[Recomposition]\label{def:recomposition}
  \[
    \forall m,
    \left.
      \begin{alignedat}{4}
        &\trg{C} & \linktrg\, & \trg{P} & \semtrg m\\
        &\trg{C'} & \linktrg\, & \trg{P'} & \semtrg m
      \end{alignedat}
    \right\}
    \implies
    \trg{C'} \linktrg \trg{P} \semtrg m.
  \]
\end{definition}
Recomposition takes executions of two target-level whole programs that produce the same trace, and requires proving that one can ``recompose'' these programs, taking the context of one and the partial program of the other and linking them, to obtain a \emph{recomposed execution} of the \emph{recomposed program} that produces the same trace.
Applying recomposition at this point yields \(\cmp[\srclang][\trglang]{\src{C_{\srclang}}} \linktrg \cmp[\srclang][\trglang]{\src{P}} \semtrg m\), which can be sent back to the source language using another compiler correctness theorem:
\begin{proposition}[Compiler Correctness (backward direction)]\label{hyp:bcc}
  \begin{align*}
    &\forall\,\src{P}\in\src{\mathbb{P}_{\srclang}}, \src{C}\in\src{\mathbb{C}_{\srclang}},m,\
     \cmp[\srclang][\trglang]{\src{C}} \linktrg \cmp[\srclang][\trglang]{\src{P}} \semtrg m \implies\\
    &\quad\exists\,m',\src{C} \linksrc \src{P} \semsrc m' \wedge
      (m \leq m' \vee m' \prec m)
  \end{align*}
  where \(m' \prec m\) means that \(m\) refines an undefined behaviour happening at the end of \(m'\), but with no indication of where the undefined behaviour happened.
\end{proposition}
Using the last proposition, blame, we finally obtain \(\rscdcmd\):
\begin{proposition}[Blame]\label{hyp:blame}
  \begin{align*}
    &\forall\,\src{P},\src{P'}\in\src{\mathbb{P}_{\srclang}}, \src{C}\in\src{\mathbb{C}_{\srclang}},m,\\
    &\quad\src{C}\linksrc\src{P} \semsrc m \wedge \src{C}\linksrc\src{P'}\semsrc m' \wedge m' \prec m \implies m'\prec_{\mathbb{P}} m.
  \end{align*}
\end{proposition}
Blame allows us to identify the component that triggers the undefined behaviour, as required by \(\rscdcmd\).

When all these theorems hold, we obtain \(\rscdcmd\).
However, as we hinted at previously, recomposition may not always hold as stated here.
To explain how this may happen, and how we solve it, we must first explain how recomposition is proved.

\subsection{Recomposition and its Failings in Low-level Settings}\label{sec:new-approach:failings}

So, how do Abate \etal~\cite{Abate_2018_arxiv} prove recomposition?

\begin{figure}[ht]
  \begin{subfigure}{0.48\columnwidth}
\begin{minted}[fontsize=\small]{c}
C1.f() {
  if (C2.g(27) > 1) {
    E.print(1);
  } else {
    E.print(0);
  }
}

\end{minted}
\begin{minted}[fontsize=\small]{c}
C2.g(x) {
  E.print(x);
  return 12;
}
\end{minted}
    \caption{First program \texttt{W}}\label{fig:rec-ok-example-prog1}
  \end{subfigure}\hfill
  \begin{subfigure}{0.48\columnwidth}
\begin{minted}[fontsize=\small]{c}
C1.f() {
  r = C2.g(27);
  if (r > 1) {
    E.print(1);
  } else {
    E.print(0);
  }
}
\end{minted}
\begin{minted}[fontsize=\small]{c}
C2.g(x) {
  E.print(27);
  return 12;
}
\end{minted}
    \caption{Second program \code{W'}}\label{fig:rec-ok-example-prog2}
  \end{subfigure}
  \caption{Two programs that recompose well}\label{fig:rec-ok-example}
\end{figure}
Let us look at a concrete example where recomposition works well (\Cref{fig:rec-ok-example}).
This example is written in a hypothetical language whose syntax is similar to C, and which has a block-based memory model like the target of Abate \etal.
The programs \texttt{W} and \code{W'} are divided into two compartments, \code{C1} with one procedure \code{C1.f} (written \code{C1.f'} in \code{W'}) and \code{C2} with one procedure \code{C2.g} (written \code{C2.g'} in \code{W'}), and we have an \code{E.print} external function that outputs values on the trace.

Let us consider what happens in \texttt{W} in an execution starting in \code{C1.f}: first, \code{C1.f} calls \code{C2.g} with argument \code{27}; \code{C2.g} calls \code{E.print} to output it, then returns \code{12} to \code{C1.f}.
\code{C1.f} then branches on the returned value, and since \code{12 > 1}, takes the first branch and calls \code{E} to print it.

The second program \code{W'} does the same thing, with the sole exception being that \code{C1.f'} uses a temporary variable for storing the result of \code{C2.g'}, and that \code{C2.g'} hardcodes the value to be printed, \code{27}.
The event traces can be represented visually as the following sequences of transitions, where \(\rightarrow^{*}\) represents a sequence of silent steps, and the labels correspond to calls, returns, or outputs via \code{E.print}.
\tikzset{
  exec st/.style={circle,fill,inner sep=0pt,minimum size=2pt},
  exec tr/.style={->,shorten >=2pt,shorten <=2pt},
  exec lb/.style={above,inner sep=1.5pt,font=\ttfamily\scriptsize},
  exec star/.style={pos=1,anchor=south east,inner sep=0pt,xshift=1pt,yshift=0.5pt,font=\scriptsize}}%
\newcommand{\execrow}[4]{
  \node[left,font=\footnotesize] at (-0.2,#3) {\code{#2}};
  \foreach \i/\x in {0/0, 1/0.55, 6/5.3, 7/5.85, 8/6.5}
    \node[exec st] (p#1-\i) at (\x,#3) {};
  \foreach \i/\x in {2/2.0, 3/2.55, 4/3.3, 5/3.85}
    \node[exec st] (p#1-\i) at (\x,#4) {};
  \foreach \a/\b in {0/1, 2/3, 4/5, 6/7}
    \draw[exec tr] (p#1-\a) -- node[exec star] {\(*\)} (p#1-\b);
  \draw[exec tr] (p#1-1) -- node[exec lb,sloped] {C2.g(27)} (p#1-2);
  \draw[exec tr] (p#1-3) -- node[exec lb] {27} (p#1-4);
  \draw[exec tr] (p#1-5) -- node[exec lb,sloped] {Ret 12} (p#1-6);
  \draw[exec tr] (p#1-7) -- node[exec lb] {1} (p#1-8);
}%
\[
  \begin{tikzpicture}
    \execrow{1}{W}{0}{0}
    \execrow{2}{W'}{-0.9}{-0.9}
  \end{tikzpicture}
\]

Now, recomposition states that the two programs can be recomposed: one can build a program \texttt{W½} by taking \code{C1.f} from \texttt{W} and \code{C2.g'} from \code{W'}.
If we follow the execution of that program, we see that it produces the exact same trace, despite having different code from \texttt{W} and \code{W'}.

The reason is that when executing inside \code{C1.f}, \texttt{W½} takes the exact same steps as \texttt{W}, and when executing inside \code{C2.g'}, the exact same steps as \code{W'}, which can be represented as the diagram below:
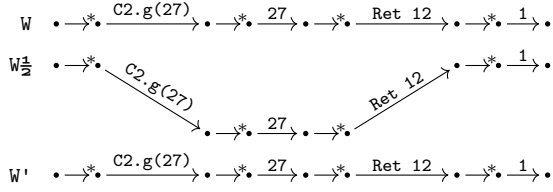
\[
  \begin{tikzpicture}
    \execrow{1}{W}{0}{0}
    \execrow{h}{W½}{-0.55}{-1.45}
    \execrow{2}{W'}{-2}{-2}
  \end{tikzpicture}
\]

In fact, this diagram is exactly how the proof of recomposition goes for both Abate \etal~\cite{Abate_2018_arxiv} and Thibault \etal~\cite{ThibaultBLAAGHT24}: they prove a \emph{three-way simulation} between the two original programs and the recomposed one.
More precisely, they prove a series of simulation diagrams capturing the idea that silent steps in the parts that are kept are matched by equivalent silent steps (steps closer together between \code{W}/\code{W'} and \code{W½} in the diagram), that silent steps in the parts that are discarded do not matter (steps further apart in the diagram), and that event-producing steps are taken synchronously.

\begin{figure}[ht]
  \begin{subfigure}[b]{0.48\columnwidth}
\begin{minted}[fontsize=\small,escapeinside=||]{c}
C1.f() {
  if (C2.g(27) > 1) {
    E.print(1);
  } else {
|\llap{100:\ }|    E.print(0);
  }
}
\end{minted}
\begin{minted}[fontsize=\small]{c}

C2.g(x) {
  E.print(*50);
  if (*100 < 1000) {
    return 0;
  } else {
    return 12;
  }
}
\end{minted}
    \caption{First program \texttt{W}}\label{fig:rec-notok-example-prog1}
  \end{subfigure}\hfill
  \begin{subfigure}[b]{0.48\columnwidth}
\begin{minted}[fontsize=\small,escapeinside=||]{c}
C1.f() {
  r = C2.g(27);
  if (r > 1) {
|\llap{100:\ }|    E.print(1);
  } else {
    E.print(0);
  }
}
\end{minted}
\begin{minted}[fontsize=\small]{c}
C2.g(x) {
  E.print(27);
  if (*100 < 1000) {
    return 12;
  } else {
    return 0;
  }
}
\end{minted}
    \caption{Second program \code{W'}}\label{fig:rec-notok-example-prog2}
  \end{subfigure}
  \caption{Two programs that don't recompose well}\label{fig:rec-notok-example}
\end{figure}
Now, let us look at what could go wrong when we introduce lower-level observations.
Consider the programs in \Cref{fig:rec-notok-example}.
The functions \code{C1.f} and \code{C1.f'} are unchanged, but \code{C2.g} and \code{C2.g'} have changed.
In \texttt{W}, \code{C2.g} first reads the value to print at an arbitrary address (\code{50}), then reads at another arbitrary address (\code{100}), and chooses which value to return based on what it finds there.
In \code{W'}, \code{C2.g'} only does this second read (and always prints \code{27}), but the branches are swapped.

In the language from the example, reading memory at an arbitrary integer address is undefined behaviour, which is modelled by the execution stopping there, as depicted in the following diagram, in black.
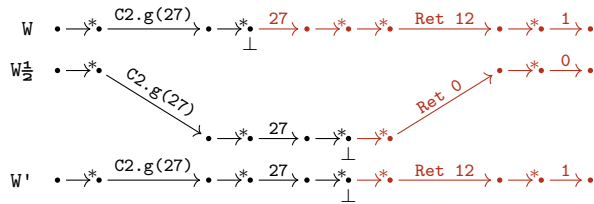
\NewDocumentCommand{\ubexecrow}{s O{black} m m m m}{
  \node[left,font=\footnotesize] at (-0.2,#5) {\code{#4}};
  \foreach \i/\x in {0/0, 1/0.55}
    \node[exec st] (q#3-\i) at (\x,#5) {};
  \foreach \i/\x in {2/2.0, 3/2.55}
    \node[exec st] (q#3-\i) at (\x,#6) {};
  \foreach \a/\b in {0/1, 2/3}
    \draw[exec tr] (q#3-\a) -- node[exec star] {\(*\)} (q#3-\b);
  \draw[exec tr] (q#3-1) -- node[exec lb,sloped] {C2.g(27)} (q#3-2);
  \IfBooleanF{#1}{
    \foreach \i/\x in {4/3.3, 5/3.85}
      \node[exec st,#2] (q#3-\i) at (\x,#6) {};
    \draw[exec tr,#2] (q#3-3) -- node[exec lb] {27} (q#3-4);
    \draw[exec tr,#2] (q#3-4) -- node[exec star] {\(*\)} (q#3-5);
  }
}%
\newcommand{\lowext}[6][BrickRed]{
  \begin{scope}[#1]
    \node[exec st] (q#2-6) at (4.4,#4) {};
    \node[exec st] (q#2-7) at (5.85,#3) {};
    \node[exec st] (q#2-8) at (6.4,#3) {};
    \node[exec st] (q#2-9) at (7.05,#3) {};
    \draw[exec tr] (q#2-5) -- node[exec star] {\(*\)} (q#2-6);
    \draw[exec tr] (q#2-6) -- node[exec lb,sloped] {#5} (q#2-7);
    \draw[exec tr] (q#2-7) -- node[exec star] {\(*\)} (q#2-8);
    \draw[exec tr] (q#2-8) -- node[exec lb] {#6} (q#2-9);
  \end{scope}
}%
\[
  \begin{tikzpicture}
    \ubexecrow[BrickRed]{1}{W}{0}{0}
    \ubexecrow{h}{W½}{-0.55}{-1.45}
    \ubexecrow{2}{W'}{-2}{-2}
    \node[below,inner sep=3pt,font=\scriptsize] at (q1-3) {\(\bot\)};
    \foreach \r in {h,2}
      \node[below,inner sep=3pt,font=\scriptsize] at (q\r-5) {\(\bot\)};
    \lowext{1}{0}{0}{Ret 12}{1}
    \lowext{h}{-0.55}{-1.45}{Ret 0}{0}
    \lowext{2}{-2}{-2}{Ret 12}{1}
  \end{tikzpicture}
\]
Since the executions of the two programs end early, recomposition only requires recomposing \texttt{W} and \code{W'} up to the last shared event (the call to \code{C2.g}), and indeed \texttt{W½} proceeds similarly up to that point.

However, after compiling further down to a flat memory model, the programs are in a lower-level language where code is laid out in memory and any integer can be used as a pointer, thus reading at an arbitrary address isn't undefined behaviour anymore.\footnote{We do not give out the code of the compiled programs explicitly, and for the purpose of this example, one can imagine the same syntax to be executed using a lower-level semantics.}
To illustrate the problems encountered when trying to prove recomposition at this lower level, let us assume that the compiler put the code of the compartments first in memory, and in particular that address \code{100} points to a particular instruction from compartment \code{C1}: in \code{W}, to the one responsible for calling \code{E.print(0)}, and in \code{W'}, to the one responsible for calling \code{E.print(1)}; and let us assume that address \code{50} in \code{W} contains the integer \code{27}.

The execution of \code{W} starts the same as in the source, but instead of encountering undefined behaviour, the reads at arbitrary addresses succeed and the execution can continue, represented in red.
\code{C2.g} first reads \code{27}, prints it, then reads from address \code{100}, obtains something that doesn't satisfy the condition, and ends up returning \code{12}.

Now, \code{W'} behaves similarly: it prints \code{27}, as before, then reads from address \code{100}, obtaining this time something that satisfies the condition and returns \code{12}.
Hence, \code{W} and \code{W'} end up having the same event trace.
Yet, if we try to recompose them, we run into an issue: \code{C2.g'} reads from address \code{100}, gets a value that doesn't pass the test, and so returns \code{0}.
Hence, \texttt{W½}'s execution extends the previous one with what's depicted in red, which is clearly not the same trace: the simulation falls apart; recomposition cannot hold.

This doesn't mean \(\rscdcmd\) doesn't hold anymore; indeed, the observations that break recomposition are observations about the code of other compartments, and, as explained by Abate \etal~\cite{AbateBGHPT19}, criteria such as \(\rsp\) and its derivatives (including \(\rscdcmd\)) do not depend on the privacy of code.
Moreover, in the proof of \(\rscdcmd\), the kept function \code{C2.g'} can only come from two places: it is either part of the compiled partial program or part of the compiled back-translation.
In the former case, that means the partial program exhibited undefined behaviour, and so we can use the blame conjunct of \(\rscdcmd\) (\(m' \prec_{\mathbb{P}} m\)) to stop early in reproducing the trace; in the latter case, \code{C2.g'} comes from a source program that the back-translation built to produce the whole trace \(m\); since reading address \code{100} is undefined behaviour in the source, that read can only happen after \(m\) ends, so it cannot affect the recomposed execution along \(m\).
Similarly, we're discarding \code{C2.g}, so what it actually does shouldn't matter at all.

This gives us a glimpse of the solution: our goal becomes to build a variant of recomposition that can be proved, because it excludes these cases, and can still be used to obtain \(\rscdcmd\).

\subsection{Overlay Semantics}\label{sec:new-approach:overlay}

Our methodology starts by introducing \emph{overlay semantics}~\cite{SkorstengaardDB19,GeorgesTB22} that prevent these undesirable behaviours: a second small-step semantics for \(\trglang\) programs, using the exact same syntax for programs, but with runtime states carrying additional information, called \emph{ghost state}.
The difference is that the small-step relation for the overlay semantics, \(\trg{\xrightarrow{\com{e}}_{\mathrm{ov}}}\), is made more restrictive by adding conditions that prevent these behaviours, using the ghost state.
We can then define \(\trg{\xrightarrow{\com{t}}}_{\trg{\mathrm{ov}}}^*\) and \({\semtrg}_{\trg{\mathrm{ov}}}\) similarly to those for the low-level semantics.

Because the overlay semantics is defined to be more restrictive than the low-level semantics, we trivially have:
\[
  \forall\,\trg{s},\trg{s'}, \trg{s} \mathrel{\trg{\xrightarrow{\com{e}}_{\mathrm{ov}}}} \trg{s'} \implies \trg{s} \mathrel{\trg{\xrightarrow{\com{e}}}} \trg{s'}
\]
and, as a consequence:
\begin{align*}
  \forall\,\trg{s},\trg{s'}, \trg{s} \mathrel{\trg{\xrightarrow{\com{e}}}_{\trg{\mathrm{ov}}}^*} \trg{s'} &\implies \trg{s} \mathrel{\trg{\xrightarrow{\com{e}}}^*} \trg{s'} \, \mathrm{ and}\\
  \forall\,\trg{W}, m, \trg{W} \mathrel{{\semtrg}_{\trg{\mathrm{ov}}}} m &\implies \trg{W} \semtrg m \text{ for } m \text{ free of UB}
\end{align*}

In the other direction, there are cases where \(\trg{s} \mathrel{\trg{\xrightarrow{\com{e}}}} \trg{s'}\) holds but \(\trg{s} \mathrel{\trg{\xrightarrow{\com{e}}_{\mathrm{ov}}}} \trg{s'}\) doesn't hold (for any \(\trg{s'}\)).
We can illustrate that fact by thinking of our previous example.
In the low-level semantics, the states \(\trg{s}\) contain a memory \(\trg{m}\), \ie a map from \(\trg{\mathtt{word}}\) to \(\trg{\mathtt{word}}\), and loading from memory is essentially unconstrained:
\begin{mathpar}
  \inferrule[Low-Level-Load]{
    \trg{m(w) = w'} \\
    \dots
  }{
    \trg{\mathop{Load}(m, w) = w'}
  }
\end{mathpar}
One can load from a location whenever it contains something, no matter what.\footnote{Other conditions, such as not being allowed to read from another compartment's data, are omitted, since they are not relevant to the discussion here, which is about distinguishing code from data in memory.}

The overlay semantics adds a bit of ghost state, here a rudimentary tag on every memory location recording whether it contains code, and uses it to restrict execution:
\begin{mathpar}
  \inferrule[Overlay-Load]{
    \trg{m(w) = w'}\\
    \highlight{\trg{m \vdash w : data}} \\
    \dots
  }{
    \trg{\mathop{Load}(m, w) = w'}
  }
\end{mathpar}
Here, the notation ``\(\trg{m \vdash w : data}\)'' means that location \(\trg{w}\) contains data, to be contrasted to the notation ``\(\trg{m \vdash w : code}\)'' that would appear on locations containing code.
According to this rule, the reads performed by \code{C2.g} and \code{C2.g'} wouldn't be allowed, because they target the code area of memory.

In summary, our methodology starts by defining an overlay semantics that gets stuck when encountering the behaviours that break recomposition, essentially seeing the low-level semantics as \emph{refining undefined behaviour} from the overlay semantics.
Defining such overlay semantics is not always easy, and one can easily imagine settings where doing so would be prohibitively complex, or even impossible.
As we will see in \Cref{sec:backend-change}, which describes our concrete implementation of the proof technique, we propose combining our methodology with concrete enforcement at the lowest level, to eliminate some behaviours and stay in a setting where defining the overlay semantics is realistically feasible. 

Now that we have defined the overlay semantics, we can reduce the problem to proving two instances of \(\rscdcmd\) and composing them: (1) for the compiler from \(\srclang\) to \(\trglang\) with the overlay semantics, and (2) for the identity compiler from \(\trglang\) with the overlay semantics to \(\trglang\) with the low-level semantics.

Because we eliminated all behaviours that invalidate recomposition in the overlay semantics, we can apply the proof technique of Abate \etal to prove (1), and this doesn't require much additional work: by definition, we can still reuse the back-translation (traces in the low-level semantics can only become smaller in the overlay semantics) and compiler correctness (the overlay semantics only stops behaviours that were anyway not allowed in the source language).
We will also see that the recomposition variant used to prove (2) also implies recomposition for the overlay semantics, so we need not duplicate this proof.
We now focus on proving (2).

\subsection{Undefined-behaviour-restricting semantics}

In order to prove \(\rscdcmd\) between the two semantics for \(\trglang\), we introduce a novel variant of recomposition; but before we can do so, we need to build a third semantics that captures the high-level intuition of ``undefined behaviour not happening in a certain program part''.
The reason we want to do this is that our proofs are not manipulating arbitrary low-level objects: one component is indeed a low-level component, but the other one is obtained by compiling a source program, and so it has defined behaviour according to the overlay semantics, at least along the trace that is considered.
This semantics should then allow us to reason about the linking of these two objects.

We call the semantics we define \emph{undefined-behaviour-restricting semantics} (UB-restricting semantics for short), because it allows some behaviours from the low-level semantics that are not allowed in the overlay semantics, but not all.
We build this by combining the overlay semantics and the low-level semantics, at the small-step level.

The semantics is parameterized by \(\trg{\mathbb{I}}\), which can be either \(\trg{\mathbb{C}}\) or \(\trg{\mathbb{P}}\), describing on which side undefined behaviours are allowed, and defined by the small-step relation \(\trg{\rightarrow_{\UB(\mathbb{I})}}\) with the following two rules:
\begin{mathpar}
  \inferrule[S-In-I]{
    \trg{s} \in \trg{\mathbb{I}} \\
    \trg{s} \mathrel{\trg{\xrightarrow{\com{e}}}} \trg{s'}
  }{
    \trg{s} \mathrel{\trg{\xrightarrow{\com{e}}_{\UB(\mathbb{I})}}} \trg{s'}
  }
  \and
  \inferrule[S-NotIn-I]{
    \trg{s} \notin \trg{\mathbb{I}} \\
    \trg{s} \mathrel{\trg{\xrightarrow{\com{e}}_{\mathrm{ov}}}} \trg{s'}
  }{
    \trg{s} \mathrel{\trg{\xrightarrow{\com{e}}_{\UB(\mathbb{I})}}} \trg{s'}
  }
\end{mathpar}
Rule \textsc{S-In-I} says that any low-level behaviour is allowed in the UB-restricting semantics, for runtime states belonging to \(\mathbb{I}\).
The other rule, \textsc{S-NotIn-I}, says that only behaviours that are possible in the overlay semantics are allowed in the UB-restricting semantics, for runtime states not belonging to \(\mathbb{I}\).

Since \(\trg{s} \mathrel{\trg{\xrightarrow{\com{e}}_{\mathrm{ov}}}} \trg{s'}\) implies \(\trg{s} \mathrel{\trg{\xrightarrow{\com{e}}}} \trg{s'}\), the UB-restricting semantics accepts any behaviour that can be expressed in the overlay semantics; but, additionally, it accepts low-level behaviours that are not possible in the overlay semantics (\ie undefined behaviours \wrt the overlay semantics) as long as they happen in the right component.

Coming back to the example from \Cref{fig:rec-notok-example}, we can execute the first program \texttt{W} with the UB-restricting semantics, in two ways depending on which compartment we assign to \(\mathbb{I}\):
\[
  \begin{tikzpicture}
    \ubexecrow*{a}{\(\mathbb{I} \ni \mathtt{C1}\)}{0}{0}
    \node[below,inner sep=3pt,font=\scriptsize] at (qa-3) {\(\bot\)};
    \ubexecrow[BrickRed]{b}{\(\mathbb{I} \ni \mathtt{C2}\)}{-0.9}{-0.9}
    \lowext{b}{-0.9}{-0.9}{Ret 12}{1}
  \end{tikzpicture}
\]
In the first execution, \code{C2} doesn't belong to \(\mathbb{I}\), so it must follow the overlay semantics, and triggers undefined behaviour when attempting to read code.
On the other hand, when \code{C2} belongs to \(\mathbb{I}\), in the second execution, these steps are allowed because the UB-restricting semantics does not constrain steps taken by compartments in \(\mathbb{I}\).

\subsection{Blame-aware recomposition}
Using these definitions, we can now formulate a novel variant of recomposition, which we call \emph{blame-aware recomposition}, that captures the ideas necessary to make the proof work in our setting.
As we mentioned previously, there is no hope of proving recomposition (\Cref{def:recomposition}).
However, we repeat the two observations that guide the definition of our new property:
\begin{itemize}
  \item \(\trg{C}\) (and \(\trg{P'}\)) are not used in the conclusion of recomposition, so in principle they may exhibit the recomposition-breaking behaviours of \Cref{sec:new-approach:failings}.
  \item Recomposition is used in a proof technique where \(\trg{C'}\) and \(\trg{P'}\) are the compiled back-translation of \(m\), so they do not exhibit the unwanted behaviours along \(m\); and where \(\trg{P}\) is compiled from a source program, so it can only exhibit them where the source program already has undefined behaviour, which the blame conjunct of \(\rscdcmd\) allows.
\end{itemize}

These observations lead us to the definition of blame-aware recomposition,
trying to capture these intuitions by using the UB-restricting semantics:
\begin{definition}[Blame-aware recomposition]\label{def:ba-recomp}
\[
  \forall m,
  \left.
    \begin{alignedat}{5}
      &\trg{C} & \linktrg\, & \trg{P} && \mathrel{\semtrg}_{\trg{\UB(\mathbb{C})}} & m\\
      &\trg{C'} & \linktrg\, & \trg{P'} && \ovsemtrg & m
    \end{alignedat}
\right\}
\implies
\trg{C'} \linktrg \trg{P} \ovsemtrg m
\]
\end{definition}

This property is similar to the recomposition property from Abate \etal~\cite{Abate_2018_arxiv}, in that it allows putting together two similar executions from different programs to obtain a third execution.
However, the difference with the previous version of recomposition lies in the semantics that are being used: \(\trg{C'} \linktrg \trg{P'}\) must execute according to the overlay semantics, while \(\trg{C} \linktrg \trg{P}\) must execute according to the UB-restricting semantics.
In the conclusion, we obtain an execution in the overlay semantics, which can trivially be lowered to the low-level semantics.
The use of the UB-restricting semantics for \(\trg{C} \linktrg \trg{P}\) means that \(\trg{C}\) is allowed to have behaviours that are not allowed in the overlay semantics, \ie it is allowed to perform the low-level behaviours that would invalidate standard recomposition.

To illustrate this definition, let us take the example from \Cref{fig:rec-notok-example} and instantiate the components as follows:
\[
  \begin{alignedat}{10}
    \mathtt{C1} &\in \mathbb{P} \qquad && \mathtt{W} &=& \, \trg{C} & \linktrg \, & \trg{P} &=& \, \mathtt{C2} \, & \linktrg \, & \mathtt{C1} \\
    \mathtt{C2} &\in \mathbb{C} \qquad && \text{\code{W'}} &=& \, \trg{C'} & \linktrg \, & \trg{P'} &=& \, \mathtt{C2'} \, & \linktrg \, & \mathtt{C1'}
  \end{alignedat}
\]
The recomposition diagram for this instantiation is:
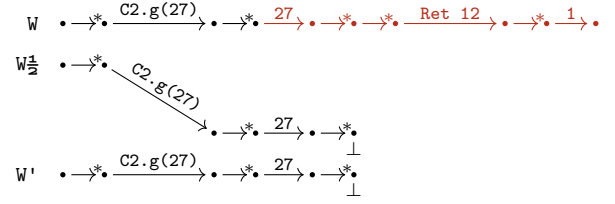
\[
  \begin{tikzpicture}
    \ubexecrow[BrickRed]{1}{W}{0}{0}
    \ubexecrow{h}{W½}{-0.55}{-1.45}
    \ubexecrow{2}{W'}{-2}{-2}
    \foreach \r in {h,2}
      \node[below,inner sep=3pt,font=\scriptsize] at (q\r-5) {\(\bot\)};
    \lowext{1}{0}{0}{Ret 12}{1}
  \end{tikzpicture}
\]
Since \(\trg{C} \linktrg \trg{P}\) executes according to the UB-restricting semantics (\(\semtrg_{\trg{\UB(\mathbb{C})}}\)), the low-level observations performed by \code{C2.g} go through, depicted in red.
However, those are not allowed in \(\trg{C'} \linktrg \trg{P'}\), so the execution still gets stuck in \code{C2.g'}, when reading at address \code{100}.
But now, note how we recompose: despite \code{W} technically triggering undefined behaviour before printing \code{27}, we are still able to include this event in the recomposed execution, because inside \code{W}, it happens in the context that will be discarded!

We also introduce an alternative formulation of blame-aware recomposition that takes into account what happens when the execution gets stuck while executing inside the program part \(\mathbb{P}\) in the UB-restricting semantics: the recomposed execution also gets stuck while executing inside the program part.
\begin{definition}[Blame-aware recomposition with explicit undefined behaviour]\label{def:ba-recomp-ub}
\[
  \forall m,
  \left.
    \begin{alignedat}{4}
      &\trg{C} &\linktrg\, & \trg{P} \mathrel{\semtrg}_{\trg{\UB(\mathbb{C})}} & &m \cdot \bot_{\mathbb{P}}\\
      &\trg{C'} &\linktrg\, & \trg{P'} \mathrel{\semtrg}_{\trg{\mathrm{ov}}} & &m 
    \end{alignedat}
\right\}
\implies
\trg{C'} \linktrg \trg{P} \mathrel{\semtrg}_{\trg{\mathrm{ov}}} m \cdot \bot_{\mathbb{P}}
\]
\end{definition}

The purpose of this definition is to properly handle the undefined behaviour happening in the program, and to link it directly to the blame conjunct of \(\rscdcmd\).
To illustrate it, let us take the other assignment (\code{C1} becomes the context and \code{C2} becomes the program), and the other recomposition (we build \code{W½'} with \code{C1} from \code{W'} and \code{C2} from \code{W}).
\[
  \begin{alignedat}{10}
    \mathtt{C1} &\in \mathbb{C} \qquad && \mathtt{W} &=& \, \trg{C} & \linktrg \, & \trg{P} &=& \, \mathtt{C1} \, & \linktrg \, & \mathtt{C2} \\
    \mathtt{C2} &\in \mathbb{P} \qquad && \text{\code{W'}} &=& \, \trg{C'} & \linktrg \, & \trg{P'} &=& \, \mathtt{C1'} \, & \linktrg \, & \mathtt{C2'}
  \end{alignedat}
\]
  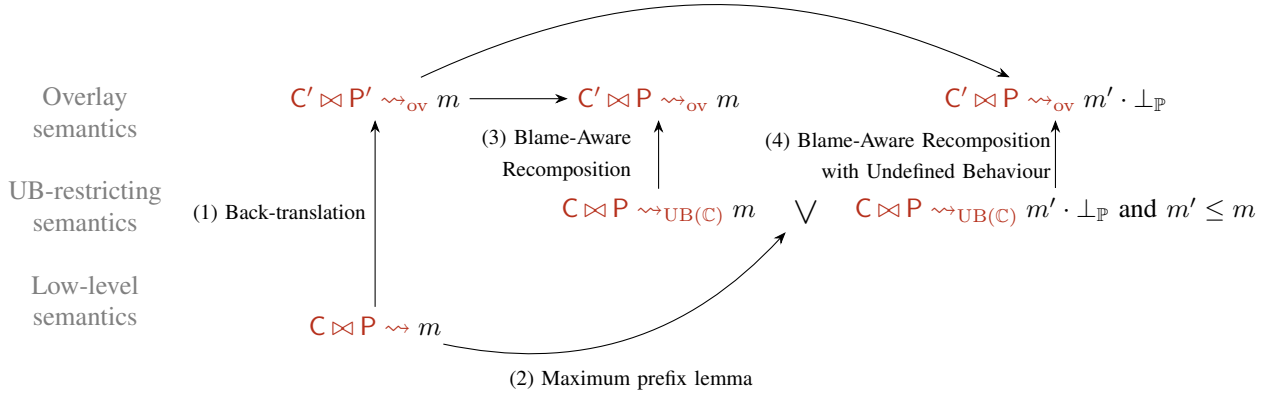
\begin{figure*}
    \begin{tikzpicture}[>=Stealth,x=3cm,y=1.5cm]

      \node (A) at (0, 0) {$\trg{C} \linktrg \trg{P} \semtrg m$};
      \node (B) at (0, 1) {};
      \node (C) at (0, 2) {$\trg{C'} \linktrg \trg{P'} \mathrel{{\semtrg}_{\trg{\mathrm{ov}}}} m$};

      \node (D) at (1.25, 1) {$\trg{C} \linktrg \trg{P} \mathrel{{\semtrg}_{\trg{\UB(\mathbb{C})}}} m$};
      \node (E) at (1.25, 2) {$\trg{C'} \linktrg \trg{P} \mathrel{{\semtrg}_{\trg{\mathrm{ov}}}} m$};

      \node (F) at (3, 1) {$\trg{C} \linktrg \trg{P} \mathrel{{\semtrg}_{\trg{\UB(\mathbb{C})}}} m' \cdot \bot_{\mathbb{P}}$ and $m' \le m$ };
      \node (G) at (3, 2) {$\trg{C'} \linktrg \trg{P} \mathrel{{\semtrg}_{\trg{\mathrm{ov}}}} m' \cdot \bot_{\mathbb{P}}$};

      \node (O) at (1.90, 1) {\Large$\vee$};

      \node[draw=none,align=center,text height=0pt,text depth=0pt] (LabA) [left=2cm of A] {\textcolor{gray}{Low-level}\\ \textcolor{gray}{semantics}};
      \node[draw=none,align=center,text height=0pt,text depth=0pt] (LabB) [above=of LabA] {\textcolor{gray}{UB-restricting}\\ \textcolor{gray}{semantics}};
      \node[draw=none,align=center,text height=0pt,text depth=0pt] (LabC) [above=of LabB] {\textcolor{gray}{Overlay}\\ \textcolor{gray}{semantics}};

      \draw[->] (A) -- node[left] {\footnotesize(1) Back-translation} (C);
      \draw[->, bend right] (A) to node[below=0.3cm] {\footnotesize(2) Maximum prefix lemma} (O);

      \draw[->] (D) -- (E);
      \draw[->] (C) -- (E);

      \draw[->] (F) -- (G);

      \node[align=right] (LabelRec1) at (0.80, 1.5) {\footnotesize (3) Blame-Aware\\ \footnotesize Recomposition};
      \node[align=right] (LabelRec2) at (2.35,1.5) {\footnotesize (4) Blame-Aware Recomposition \\ \footnotesize with Undefined Behaviour};

      \draw[->, bend left=25] (C) to (G);

    \end{tikzpicture}
    \caption{Proof technique for \(\rscdcmd\) between the overlay semantics and the low-level semantics}\label{fig:proof-technique}
  \end{figure*}
For this assignment, the recomposition diagram is:
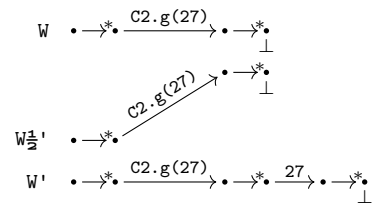
\[
  \begin{tikzpicture}
    \ubexecrow*{1}{W}{0}{0}
    \ubexecrow*{h}{W½'}{-1.45}{-0.55}
    \ubexecrow{2}{W'}{-2}{-2}
    \foreach \r in {1,h}
      \node[below,inner sep=3pt,font=\scriptsize] at (q\r-3) {\(\bot\)};
    \node[below,inner sep=3pt,font=\scriptsize] at (q2-5) {\(\bot\)};
  \end{tikzpicture}
\]
Here, the undefined behaviour that happens in the first execution happens in the program part, hence it must also happen in the program part in \code{W½'}.
As we will see in \Cref{sec:new-approach:replay}, we need this second formulation for programs \(\trg{P}\) with undefined behaviour.

Importantly, blame-aware recomposition directly implies recomposition at the level of the overlay semantics, which simplifies proving that the compiler to the overlay semantics satisfies \(\rscdcmd\), \ie instance (1) of \Cref{sec:new-approach:overlay}.

One may wonder why we chose to use the overlay semantics, and not the UB-restricting semantics, in the second execution, \(\trg{C'} \linktrg \trg{P'} \mathrel{\semtrg}_{\trg{\mathrm{ov}}} m\).
Indeed, it appears to be possible to replace it with \({\semtrg}_{\trg{\UB(\mathbb{P})}}\), since \(\trg{P'}\) is being discarded, just like \(\trg{C}\).
We chose not to, because \(\trg{C'} \linktrg \trg{P'}\) is the compiled back-translation of \(m\), so \(\trg{P'}\) has defined behaviour in the overlay semantics along \(m\) anyway.
Because this semantics is more restrictive, it is intuitively easier to reason about it, and so we prefer using it there.

\subsection{Replaying the proof}\label{sec:new-approach:replay}

Using the two formulations of blame-aware recomposition, it becomes possible to replay
the proof of Abate \etal~\cite{Abate_2018_arxiv} and obtain \(\rscdcmd\).
As mentioned in \Cref{sec:new-approach:overlay}, we only need to prove \(\rscdcmd\) for instance (2), the ``identity
compiler'' between \(\trglang\) with the overlay semantics and \(\trglang\) with
the low-level semantics.

We need one last additional lemma before proceeding: the maximum prefix lemma.
\begin{lemma}[Maximum prefix]\label{lma:max-prefix}
  \begin{align*}
    &\forall\,\trg{P},\trg{C},m, \trg{C} \linktrg \trg{P} \semtrg m \implies \\
    &\exists\,m', \trg{C} \linktrg \trg{P} \mathrel{{\semtrg}_{\trg{\UB(\mathbb{C})}}} m' \wedge (m = m' \vee m' \prec_{\mathbb{P}} m)
  \end{align*}
\end{lemma}
\begin{proof}
  If \(\trg{C} \linktrg \trg{P} \mathrel{{\semtrg}_{\trg{\UB(\mathbb{C})}}} m\),
  then we are trivially done.

  Otherwise, we know there exist states \(\trg{s}\) and \(\trg{s'}\), and a finite
  trace \(m' < m\) such that
  \(\operatorname{\trg{initial}}(\trg{C} \linktrg \trg{P}) \mathrel{\trg{\xrightarrow{\com{m'}}}^*_{\trg{\UB(\mathbb{C})}}} \trg{s}\),
  \(\trg{s} \mathrel{\trg{\xrightarrow{\com{e}}}} \trg{s'}\), but not
  \(\trg{s} \mathrel{\trg{\xrightarrow{\com{e}}}_{\trg{\UB(\mathbb{C})}}} \trg{s'}\).
  By definition of the UB-restricting semantics, \(\trg{s} \in \trg{\mathbb{P}}\),
  and so we have \(\trg{C} \linktrg \trg{P} \mathrel{{\semtrg}_{\trg{\UB(\mathbb{C})}}} m' \cdot \bot_{\mathbb{P}}\).
\end{proof}
The maximum prefix lemma allows us to lift low-level executions to the
UB-restricting semantics by making the following observation: if
\(\trg{C} \linktrg \trg{P} \semtrg m\), then there exists some prefix
\(m' \leq m\) that can also be produced in the UB-restricting semantics, and
additionally, if it's a strict prefix, then this is because \(\trg{P}\), not
\(\trg{C}\), triggered an undefined behaviour.
Depending on the case, we use one of the two formulations of blame-aware
recomposition, \Cref{def:ba-recomp} or \Cref{def:ba-recomp-ub}, respectively.

We can now prove the result we want:
\begin{theorem}[\(\rscdcmd\) between the overlay semantics and the low-level semantics]
  \begin{align*}
    &\forall\,\trg{P}\in\trg{\mathbb{P}_{\trglang}},\trg{C}\in\trg{\mathbb{C}_{\trglang}},m,\
     \trg{C} \linktrg \trg{P} \semtrg m \implies \\
    &\quad\exists\,\trg{C'}\in\trg{\mathbb{C}_{\trglang}},m',\
     \trg{C'} \linktrg \trg{P} \mathrel{{\semtrg}_{\trg{\mathrm{ov}}}} m' \wedge
      (m \leq m' \vee m' \prec_{\mathbb{P}} m).
  \end{align*}
\end{theorem}
\begin{proof}
  The proof is described in \Cref{fig:proof-technique}.
  
  We start from the hypothesis \(\trg{C} \linktrg \trg{P} \semtrg m\) and (1) apply the back-translation to it, to obtain a new target program \(\trg{P'}\) and a new target context \(\trg{C'}\) such that \(\trg{C'} \linktrg \trg{P'} \mathrel{{\semtrg}_{\trg{\mathrm{ov}}}} m\).

  To the same hypothesis, (2) we also apply the maximum prefix lemma (\Cref{lma:max-prefix}), which yields two cases: either \(\trg{C} \linktrg \trg{P} \mathrel{{\semtrg}_{\trg{\UB(\mathbb{C})}}} m\), or \(\trg{C} \linktrg \trg{P} \mathrel{{\semtrg}_{\trg{\UB(\mathbb{C})}}} m' \cdot \bot_{\mathbb{P}}\), with \(m' \prec_{\mathbb{P}} m\).

  In the first case, no undefined behaviour was triggered in \(\trg{P}\) along the execution producing \(m\).
  Hence, (3) we can apply blame-aware recomposition (\Cref{def:ba-recomp}) to this execution and the one from the back-translation, obtaining exactly the first disjunct of the conclusion of \(\rscdcmd\): \(\trg{C'} \linktrg \trg{P} \mathrel{{\semtrg}_{\trg{\mathrm{ov}}}} m\).

  In the second case, \(\trg{P}\) triggered undefined behaviour along that execution, so (4) we apply its variant with explicit undefined behaviour (\Cref{def:ba-recomp-ub}), obtaining the second disjunct directly: \(\trg{C'} \linktrg \trg{P} \mathrel{{\semtrg}_{\trg{\mathrm{ov}}}} m' \cdot \bot_{\mathbb{P}}\) with \(m' \prec_{\mathbb{P}} m\).
\end{proof}

\subsection{Formalization}
We have formalized this generic proof technique in the Rocq proof assistant.
More precisely, we have defined the lemmas for back-translation, for
compiler correctness, and for blame.
We have defined the UB-restricting semantics in a generic manner, and we
prove that in that setting, the maximum prefix lemma holds and that blame-aware recomposition can be used to prove \(\rscdcmd\).

\section{A Low-Level Micro-Policy Backend}\label{sec:backend-change}
In this section, we discuss how our proof technique can be applied
to a concrete low-level backend.

\subsection{The low-level micro-policy backend of Abate \texorpdfstring{\etal}{et al.}}
Abate \etal~\cite{AbateABEFHLPST18} implement and verify a secure compiler for
a compartmentalized C-like language.
In their source language, programs are split into several compartments
\(C\) that can only interact with each other via function calls, mediated
by interfaces that specify what procedures are exported and imported by
each compartment.
As we mentioned earlier, their compiler is only verified to be secure down to an
intermediate language with built-in support for cross-compartment calls and
returns (\ie, they are done with special instructions) and a CompCert-like
block-based memory model.
This gives them the freedom to target different backends, and they implement two
of them: one backend based on software fault isolation (SFI)~\cite{sfi_sosp1993},
and another one based on micro-policies~\cite{micropolicies2015}, which are
essentially programmable tag-based security monitors.
This second backend is the one that is of direct interest to us here and the one
we want to verify.

In that backend, the memory forms a single contiguous block and the only
datatype is that of machine words, which can act as both integers and pointers.
Additionally, on a micro-policy machine, each memory location or register is
given a tag; during execution, at each instruction, the tag of the current PC,
the tags of the operands, the tag of the current instruction, as well as the
opcode are all processed by a monitor that decides, based on these tags, whether
the execution is allowed to continue, and if so, how to update the tags.
The policy that Abate \etal implement is a compartmentalizing policy that mainly
relies on two ideas:
\begin{itemize}
  \item Memory locations are assigned a \emph{colour} that uniquely identifies the compartment they belong to. Compartments cannot read or write data to memory tagged with a colour other than their own.
  \item The stack discipline is enforced using \emph{linear return capabilities}~\cite{ieee_hst2013}. The size of the stack is recorded in the tag of the PC, and incremented every time a cross-compartment call occurs. Return addresses are tagged with a \emph{level} \(\mathtt{n}\) that represents the size of the stack at the point where the return address was created. When jumping using such an address, the policy ensures that the stack has the right size, \(\mathtt{n + 1}\). The monitor also ensures return addresses cannot be duplicated (the ``linear'' part), which prevents reusing the same return address several times.
\end{itemize}

Concretely, their backend consists of a few instructions given below, and its semantics is given by a small-step relation \(\rightarrow\), which itself depends on the compartmentalization micro-policy.
\[
  \begin{array}{lll}
    \mathtt{i} := &|\ \mathtt{Nop}\ &|\ \mathtt{Halt} \\
                  &|\ \mathtt{Const\ i \rightarrow r} &|\ \mathtt{BinOp\ r_1 \oplus r_2 \rightarrow r_d} \\
                  &|\ \mathtt{Bnz\ r\ i}  &|\ \mathtt{Mov\ r_s \rightarrow r_d} \\
                  &|\ \mathtt{Load\ {*r_p} \rightarrow r_d} &|\ \mathtt{Store\ {*r_p} \leftarrow r_s}\\
                  &|\ \mathtt{Jal\ l} &|\ \mathtt{Jump\ r} 
  \end{array}
\]

Memory addresses, registers, and the PC are tagged slightly differently, using the following tags:
\begin{align*}
  \text{Value tags }t_{\mathrm{v}} &:= \mathtt{Ret\ n}\ |\ \bot\\
  \text{Register tags }t_{\mathrm{reg}} &:= t_{\mathrm{v}}\\
  \text{Memory tags }t_{\mathrm{m}} &:= (t_{\mathrm{v}}, c, \mathtt{entry\_point})\\
  \text{PC tags }t_{\mathrm{pc}} &:= (\mathtt{Level\ n}, c).
\end{align*}

The most basic tags, upon which the others are constructed, are \emph{value
  tags} \(t_v\): \(\mathtt{Ret\ n}\) is used to tag return addresses, and
\(\bot\) for other values.
The PC is tagged with a level, which records the size of the stack, and the current compartment \(c\).
Memory locations are tagged using a value tag, to track whether they contain a return address or not; the colour \(c\) is a compartment identifier that identifies which compartment owns the memory location; and \(\mathtt{entry\_point}\) optionally records whether this is an address corresponding to the entry point of a function available to other compartments.

These tags are manipulated according to the rules of the monitor such as the ones in \Cref{fig:monitor-rules}.
Rules are of the form \(\left[P\right]\ \mathtt{Instr}\ \left[Q\right]\), meaning that \(\mathtt{Instr}\) can only be executed when the precondition \(P\) holds, and updates tags according to postcondition \(Q\).
Those conditions are of the form \(d \mapsto t\) where \(d\) is either a register \(r\), the address contained in a register \(*r\), or a code label \(l\), and \(t\) is a tag of the appropriate kind, with some conditions on it.
We use an underscore \(\_\) either when there's no condition on part of the tag (in a precondition), or when that part of the tag isn't modified (in a postcondition).

\begin{figure}[ht]
  \begin{mprules}
    \mprule{Jal\ l}
    {\mathtt{l} \mapsto (\bot, c', e \in \mathtt{allowed}(c)) \\ \mathtt{PC} \mapsto (\mathtt{Level\ n}, c)}
    {\mathtt{PC} \mapsto (\mathtt{Level\ n+1}, c') \\ \mathtt{RA} \mapsto \mathtt{Ret\ n}} \\[4ex]
    \mprule{Jump\ r}
    {\mathtt{r} \mapsto \mathtt{Ret\ n} \\ \mathtt{*r} \mapsto (\_, c', \_) \\ \mathtt{PC} \mapsto (\mathtt{Level\ n+1}, c)}
    {\mathtt{PC} \mapsto (\mathtt{Level\ n}, c') \\ \mathtt{r} \mapsto \bot} \\[5.5ex]
    \mprule{Load\ {*r_p} {\rightarrow} r_d}
    {\mathtt{*r_p} \mapsto (t_{\mathrm{v}}, c, {\_}) \\ \mathtt{PC} \mapsto ({\_}, c)}
    {\mathtt{*r_p} \mapsto (\bot, c, {\_}) \\ \mathtt{r_d} \mapsto t_{\mathrm{v}}}
  \end{mprules}
  \caption{Micro-policy describing the linear return policy for Abate \etal's backend}\label{fig:monitor-rules}
\end{figure}

The three rules depicted in \Cref{fig:monitor-rules} form the core of the linear return capability policy.
The first rule, for \(\mathtt{Jal}\), describes what happens during a cross-compartment call (calls are made using the \(\mathtt{Jal}\) instruction).\footnote{There are also rules for internal jumps that do not manipulate the tags other than checking that there is no change of compartment, but we don't depict them here as they aren't relevant to the cross-compartment concerns.}
The preconditions require the monitor to check that the label points to code belonging to a compartment \(c'\) that the current compartment \(c\) (indicated by the tag of the PC: \(\mathtt{PC} \mapsto (\mathtt{Level\ n}, c)\)) is allowed to call (\(e \in \mathop{\mathtt{allowed}}(c)\)), the set \(\mathop{\mathtt{allowed}}(c)\) being computed statically based on the interfaces.
The postcondition indicates how the monitor updates the tags: the PC is now tagged with the new compartment \(c'\), and its level incremented.
The new return capability created in the RA register must then be tagged with the former PC level, \(n\).

The next rule, for \(\mathtt{Jump}\), describes what happens when doing a jump using a return capability, \ie a cross-compartment return.
The precondition describes the current tags: the current compartment is \(c\), and the stack is at level \(\mathtt{n+1}\); meanwhile, the jump must use a return capability at just one level below, \(\mathtt{n}\), and this capability points to compartment \(c'\).
The postcondition describes the updated tags after the return: the level of the PC has been decremented, we are now in compartment \(c'\), and importantly the return capability has been invalidated, ensuring it can't be used again in another return (\(\mathtt{r} \mapsto \bot\)).

Finally, the rule for \(\mathtt{Load}\) is the one that ties everything together.
The precondition requires reading at an address belonging to the same compartment \(c\) as the PC, and that it has a tag \(t_{\mathrm{v}}\) (\(\mathtt{*r_p} \mapsto (t_{\mathrm{v}}, c, \_)\)).
The postcondition transfers that tag to the destination register (\(\mathtt{r_d} \mapsto t_{\mathrm{v}}\)) and invalidates the memory tag (\(\mathtt{*r_p} \mapsto (\bot, c, {\_})\)).
If \(t_{\mathrm{v}}\) was already \(\bot\), this doesn't change anything.
However, if \(t_{\mathrm{v}}\) was \(\mathtt{Ret\ n}\) for some \(\mathtt{n}\), then the rule ensures that there only ever exists one copy of that capability, hence the name ``linear''.

The step relation \(\rightarrow\) uses the monitor to check and modify the tags at every step. Otherwise, it is unsurprising, so it is not reproduced here.

\subsection{Why recomposition doesn't work}\label{sec:recomp-doesnt-work}
In order to apply our proof technique to Abate \etal's backend, we start by identifying why recomposition doesn't hold, \ie what are the problematic behaviours programs can have that are not recorded on the trace, and allow compartments to learn something about other compartments.
In this backend, these behaviours fall in two categories, corresponding to the two examples from the introduction.

\emph{Direct code observations} are of two kinds.
The first corresponds to \Cref{ex:direct-read} and the example from the previous section. Compartments can freely use integers as pointers, so they read from arbitrary locations in the code, obtaining information either directly about the private implementation of other compartments, or indirectly by being able to observe the code layout.
The second kind is that compartments have access to pointers to their own code, via the RA register for instance.
By converting those to integers, they learn immediately where their code is laid out in memory, which is something that depends on the code of the other compartments.
Note that this is not simply following a code pointer: here, the code pointer is just inspected by an operation or comparison.

The other category is what we call \emph{leaky allocations}, and corresponds to \Cref{ex:alloc-observation}.
These arise from the ability to cast dynamically allocated pointers to integers, which allows compartments to obtain information covertly, and from the memory allocation mechanism.
The micro-policy backend's allocator is a standard bump allocator: it allocates memory sequentially by keeping a pointer to the start of the not-yet-allocated region, bumping it by the requested amount and tagging the region between with the colour of the compartment that allocates, which can leak information about the allocation behaviour.
Recall \Cref{ex:alloc-observation}: \code{C1.f} allocates \code{x}, calls \code{C2.g}, allocates \code{y}, and then compares them.
If \code{C2.g} doesn't allocate, we get the following memory layout:
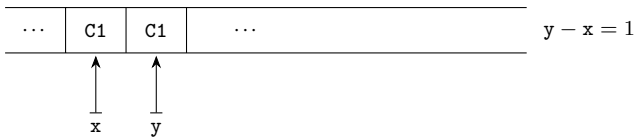
\begin{center}
  \begin{tikzpicture}[
    ptr/.style={|-{Stealth},shorten >=1pt},
    lbl/.style={font=\footnotesize}]
    \draw (-0.8,0) -- (6.1,0) (-0.8,0.6) -- (6.1,0.6) (0,0) -- (0,0.6);
    \node[lbl] at (-0.4,0.3) {\dots};
    \foreach \a/\b/\c in {0/0.8/C1, 0.8/1.6/C1} {
      \draw (\b,0) -- (\b,0.6);
      \node[lbl] at ({(\a+\b)/2},0.3) {\code{\c}};
    }
    \node[lbl] at (2.4,0.3) {\dots};
    \draw[ptr] (0.4,-0.8) node[below,lbl] {\code{x}} -- (0.4,0);
    \draw[ptr] (1.2,-0.8) node[below,lbl] {\code{y}} -- (1.2,0);
    \node[right,lbl] at (6.2,0.3) {\(\mathtt{y} - \mathtt{x} = 1\)};
  \end{tikzpicture}
\end{center}
In this case, \code{x} and \code{y} point to two locations next to each other, since nothing was allocated by \code{C2.g}. Hence, \(\mathtt{y} - \mathtt{x} = 1\).
On the other hand, if \code{C2.g} does allocate, we get the following memory layout:
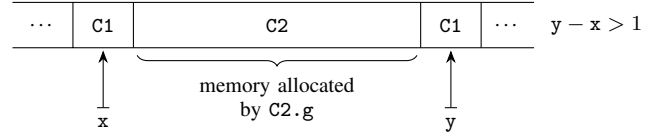
\begin{center}
  \begin{tikzpicture}[
    ptr/.style={|-{Stealth},shorten >=1pt},
    lbl/.style={font=\footnotesize}]
    \draw (-0.8,0) -- (6.1,0) (-0.8,0.6) -- (6.1,0.6) (0,0) -- (0,0.6);
    \node[lbl] at (-0.4,0.3) {\dots};
    \foreach \a/\b/\c in {0/0.8/C1, 0.8/4.6/C2, 4.6/5.4/C1} {
      \draw (\b,0) -- (\b,0.6);
      \node[lbl] at ({(\a+\b)/2},0.3) {\code{\c}};
    }
    \node[lbl] at (5.75,0.3) {\dots};
    \draw[ptr] (0.4,-0.8) node[below,lbl] {\code{x}} -- (0.4,0);
    \draw[ptr] (5,-0.8) node[below,lbl] {\code{y}} -- (5,0);
    \draw[decorate,decoration={brace,mirror,amplitude=4pt}] (0.85,-0.1) -- (4.55,-0.1)
    node[midway,below=5pt,lbl,align=center] {memory allocated\\by \code{C2.g}};
    \node[right,lbl] at (6.2,0.3) {\(\mathtt{y} - \mathtt{x} > 1\)};
  \end{tikzpicture}
\end{center}
After allocating some memory for \code{x}, the allocator has allocated a block of memory right next to it for \code{C2.g}, and then allocated again next to it for \code{y}.
Hence, this time \(\mathtt{y} - \mathtt{x} > 1\), an observation that breaks recomposition.

We reiterate that these observations being possible should not affect the top-level theorem, because the top-level theorem, as a variant of \(\rsp\), is only concerned with safety properties, and these safety properties cannot capture attacks against the privacy of the code.

We've seen in \Cref{sec:new-approach:overlay} that one solution would be to define an overlay semantics that prevents all these behaviours.
However, leaks through the allocation mechanism are much harder than code observations to characterize in an overlay semantics.
Indeed, it is difficult to accurately characterize all the observations through this kind of cast that break recomposition: doing so would already require a fine-grained dynamic analysis that tracks the provenance of pointers and how they are used, but it would also probably stop many legitimate behaviours that do not lead to leaks, thus strongly restricting the reach of our theorem.
Thus, we took a path of lesser resistance and opted to slightly modify the backend to prevent those entirely.
This is not entirely unreasonable: the backend is already preventing all kinds of undesirable behaviours, and our change is quite minimal, as it is only a slight change to the behaviour of the allocator.
Furthermore, hiding the allocation pattern would be a required change if we wanted to prove a stronger security result that requires confidentiality of data.
We describe this change in the next section.

\subsection{Hiding the memory allocation pattern in the backend}
We implement a masking protection in the allocator that stops compartments from obtaining information about other compartments' behaviour by comparing dynamically allocated pointers.
Our changes to the allocator are inspired by software-fault isolation, a commonly deployed technique to implement compartmentalization~\cite{YeeSDCMOONF10,NarayanDGFRLSS20,CastroCMPADBB09}.

Our allocator divides the memory into regions of equal size by reserving the first bits of each pointer to encode a unique compartment identifier.
This means that for a compartment \(C\), all pointers allocated by \(C\) will start with a specific bit pattern unique to that compartment.
For instance, if we assign prefix \code{0001} to \code{C1} and \code{000F} to \code{C2}, the resulting memory layout is as follows:
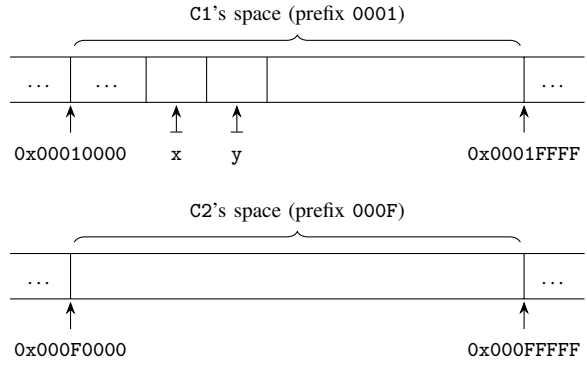
\begin{center}
  \begin{tikzpicture}[
    ptr/.style={|-{Stealth},shorten >=1pt},
    lbl/.style={font=\footnotesize,text height=1.5ex,text depth=.25ex}]
    \foreach \y/\lo/\hi/\name/\prefix in {0/0x00010000/0x0001FFFF/C1/0001, -2.6/0x000F0000/0x000FFFFF/C2/000F} {
      \begin{scope}[yshift=\y cm]
        \draw (-0.8,0) -- (6.8,0) (-0.8,0.6) -- (6.8,0.6) (0,0) -- (0,0.6) (6,0) -- (6,0.6);
        \node[lbl] at (-0.4,0.3) {\dots};
        \node[lbl] at (6.4,0.3) {\dots};
        \draw[{Stealth}-,shorten <=1pt] (0,0) -- (0,-0.4) node[below,lbl] {\code{\lo}};
        \draw[{Stealth}-,shorten <=1pt] (6,0) -- (6,-0.4) node[below,lbl] {\code{\hi}};
        \draw[decorate,decoration={brace,amplitude=4pt}] (0.05,0.75) -- (5.95,0.75)
        node[midway,above=5pt,lbl] {\code{\name}'s space (prefix \code{\prefix})};
      \end{scope}
    }
    \draw (1,0) -- (1,0.6) (1.8,0) -- (1.8,0.6) (2.6,0) -- (2.6,0.6);
    \node[lbl] at (0.5,0.3) {\dots};
    \draw[ptr] (1.4,-0.4) node[below,lbl] {\code{x}} -- (1.4,0);
    \draw[ptr] (2.2,-0.4) node[below,lbl] {\code{y}} -- (2.2,0);
  \end{tikzpicture}
\end{center}
With this scheme, \code{C1}'s allocations will only ever live between addresses \code{0x00010000} and \code{0x0001FFFF}, and \code{C2}'s allocations will only ever live between addresses \code{0x000F0000} and \code{0x000FFFFF}, ensuring they are completely disjoint.
Hence, we will always have \(\mathtt{y} - \mathtt{x} = 1\): this partitioning hides each compartment's allocation pattern.

Note that contrary to some existing SFI implementations, we do not rewrite reads using masking, and the partitioning into several regions is not a substitute for the tagging mechanism: a compartment is free to try to access any memory location, and it is the micro-policy that stops it from accessing the other compartments' data.

We do not partition the code section of the program.
The first reason for this is technical: jump labels appearing in the \(\mathtt{Jal}\) instructions are encoded as immediates that are smaller than the full range of a pointer; redirecting those to the memory space of a compartment would require substantial changes to the semantics of \(\mathtt{Jal}\), which we explicitly wanted to avoid.
Instead, we decided to put the code of all compartments contiguously at the beginning of the memory (starting with prefix \code{0000}).
The second reason is conceptual: observations that leak information about the code layout can be reasonably easily characterized as undefined behaviour in an overlay semantics, so unlike memory allocation, we don't need to change the backend for this.

\paragraph{Discussion} One may wonder whether the changes to the allocator are too large.
We do not believe this is the case.
Indeed, while such protections are not necessary to prove \(\rsp\) and its variants, including \(\rscdcmd\), they are still desirable if one is interested in confidentiality, and are necessary to prove stronger security criteria.
Moreover, protections such as using a custom allocator and masking are already commonly used~\cite{NarayanDGFRLSS20,KolosickNJWLGJS22,MorrisettTTTG12,YeeSDCMOONF10}.
Finally, it is true that our custom allocator limits the possible number of compartments, and the memory available for such compartments.
We believe this to be a necessary trade-off between security, performance, usability and proof engineering.

The remaining question is whether we could have found a proper characterization of memory allocation-related leaks as undefined behaviour for the overlay semantics.
We believe this task to be a difficult one and beyond the scope of this paper.
Indeed, in the block-based memory model of CompCert, comparing and computing with pointers is allowed when they are pointers to the same block, so preventing all pointer comparisons seems too restrictive.
Tightening this restriction would also be difficult: at the micro-policy machine level, we have already lost semantic information about which pointers used to point to the same block, and we can't retrieve it from compilation, because not all the programs we consider are compiled.
Hence, we believe this change to the allocator to be the most reasonable solution.

\section{Proving Blame-aware Recomposition}\label{sec:recomp-proofs}

We are now ready to discuss how we can apply the proof technique to this modified backend. We start by describing the overlay semantics we introduce, and then how we prove blame-aware recomposition.

\subsection{The Overlay Semantics}
Our overlay semantics uses a variant of the low-level micro-policy, which we call the \emph{overlay micro-policy}.
Its goal is to prevent the code observations described in \Cref{sec:recomp-doesnt-work}.
To do so, we add a boolean \(b\) to memory tags that records whether a location contains code.
\[
  t_{\mathrm{m}} := (t_{\mathrm{v}}, c, \mathtt{entry\_point}, b)
\]
This boolean is set statically and is completely immutable, and our overlay micro-policy uses it to prevent leaking the layout of the code.

The first major change in the overlay micro-policy is that it prevents reading code, whereas this is always allowed in the low-level micro-policy (code isn't private).
\begin{mathpar}
    \mprule{Load\ {*r_p} {\rightarrow} r_d}{\mathtt{*r_p} \mapsto (t_{\mathrm{v}}, c, {\_}, \mathtt{false}) \\ \mathtt{PC} \mapsto ({\_}, c)}{\dots}
\end{mathpar}

To help with hiding the code layout, we also make the value tags more granular, so that the machine can distinguish more types of values:
\[
  t_{\mathrm{v}} := \mathtt{Ret\ n}\ |\ \mathtt{InternalJump}\ |\ \mathtt{Inv}\ |\ \bot.
\]
The tags \(\mathtt{Ret\ n}\) and \(\bot\) are unchanged, and we introduce two new tags, \(\mathtt{InternalJump}\), for return addresses of internal calls, and \(\mathtt{Inv}\), standing for ``invalidated'', for cross-compartment return capabilities that have been used.

The invariant our overlay micro-policy maintains is that all code pointers one gets during execution (\ie return addresses) are tagged with either \(\mathtt{Ret\ n}\), \(\mathtt{InternalJump}\), or \(\mathtt{Inv}\).
Then, we simply prevent values with such tags from being used in computations or comparisons.
For instance, the \(\mathtt{Bnz}\) rule becomes the following:
\begin{mathpar}
  \mprule{Bnz\ r\ i}
  {\mathtt{r} \mapsto \bot \\
    \mathtt{PC} \mapsto ({\_}, c)\\
    \mathtt{PC + i} \mapsto ({\_}, c, {\_}, {\_}) }
  {\cdot}.
\end{mathpar}
The tag of the operand is required to be \(\bot\), meaning it can't be a pointer to code.

To tag code pointers properly, we update the micro-policy rules for the instructions that create and use them.
The only instruction that actually creates code pointers is \(\mathtt{Jal}\).
The rule for cross-compartment \(\mathtt{Jal}\) already tags the return address as \(\mathtt{Ret\ n}\); now, the rule for internal \(\mathtt{Jal}\) also tags the return address as \(\mathtt{InternalJump}\):
\begin{mathpar}
  \mprule{Jal\ l}
  {\mathtt{l} \mapsto (\bot, c, {\_}, {\_}) \\ \mathtt{PC} \mapsto ({\_}, c)}
  {\mathtt{PC} \mapsto ({\_}, c) \\ \mathtt{RA} \mapsto \mathtt{InternalJump}}.
\end{mathpar}
Because this is an internal jump, there is no need to track its level.

Now, recall the previous rule for cross-compartment \(\mathtt{Jump}\): it set the tag of the operand to \(\bot\), to prevent using it again, as required by the linear return policy.
The issue is that \(\bot\)-tagged values can still be used in comparisons and computations, so we have to change the cross-compartment \(\mathtt{Jump}\) to \emph{invalidate} it.
\begin{mathpar}
  \mprule{Jump\ r}
  {\mathtt{r} \mapsto \mathtt{Ret\ n} \\ \mathtt{*r} \mapsto ({\_}, c', {\_}, {\_}) \\ \mathtt{PC} \mapsto (\mathtt{Level\ n+1}, c)}
  {\mathtt{PC} \mapsto (\mathtt{Level\ n}, c') \\ \mathtt{r} \mapsto \mathtt{Inv}}
\end{mathpar}
That way, we ensure that the return address will not be used.
We need the same change for \(\mathtt{Load}\), \(\mathtt{Store}\), and \(\mathtt{Mov}\): they now invalidate return addresses.
Finally, we need to change the rule for internal \(\mathtt{Jump}\) to require having an \(\mathtt{InternalJump}\):
\begin{mathpar}
  \mprule{Jump\ r}
  {\mathtt{r} \mapsto \mathtt{InternalJump} \\ \mathtt{*r} \mapsto ({\_}, c, {\_}, {\_}) \\ \mathtt{PC} \mapsto ({\_}, c)}
  {\cdot}
\end{mathpar}
We need this, because absolute jumps can leak information about the code layout, just like reading the code.
Note that contrary to the cross-compartment jumps, internal jumps are not subject to the linear policy (the tags aren't modified): a compartment is free to reuse the same return address several times.

There are two additional mundane changes: the overlay micro-policy prevents executing instructions not tagged as code, and it prevents writing to code, which are needed to ensure the confidentiality of code.

Note that we didn't try to design the minimal overlay semantics: it is possible that some of these rules are not strictly necessary in our proof (for instance, preventing executing data may not be fully necessary).
However, the overlay semantics is only a tool for proving \(\rscdcmd\); these rules make reasoning about the behaviour of the code simpler, and since the source language doesn't allow executing data or rewriting code, these rules do not overly restrict the behaviours of compiled programs.

\subsection{Proving Blame-Aware Recomposition, in Practice}

Recall that our goal is to prove blame-aware recomposition:
\[
  \forall m,
  \left.
    \begin{alignedat}{5}
      &\trg{C} &\linktrg\, & \trg{P} && \mathrel{\semtrg}_{\trg{\UB(\mathbb{C})}} & m\\
      &\trg{C'} &\linktrg\, & \trg{P'} && \ovsemtrg & m
    \end{alignedat}
\right\}
\implies
\trg{C'} \linktrg \trg{P} \ovsemtrg m.
\]

\textbf{Approach}
We prove that result (and its variant with undefined behaviour) by reusing the simulation diagrams that Thibault \etal~\cite{ThibaultBLAAGHT24} introduced to prove recomposition in SECOMP and that we reproduce in \Cref{fig:recomposition-diagrams}.
Because these diagrams were originally phrased in terms of arbitrary labelled transition systems, the fact that we use different semantics in our statement is not a problem and we can apply the generic result from SECOMP in our setting.

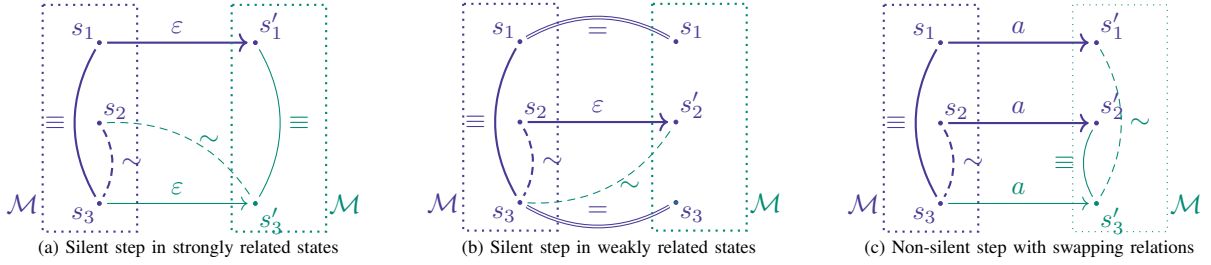
\begin{figure*}
  \centering
\begin{subfigure}{0.3\textwidth}
  \centering
  \begin{tikzpicture}
    \tikzstyle{every node}=[BlueViolet]
    \tikzstyle{every path}=[BlueViolet]
    \node[draw,shape=circle,fill=BlueViolet,inner sep=0.5pt](ini1) {};
    \node[above left = 0.1pt of ini1, inner sep=0.5pt] (textini1) {\(s_{1}\)};
    \node[draw,shape=circle,fill=BlueViolet,inner sep=0.5pt,right = 2cm of ini1] (after1) {};
    \node[above right = 0.1pt of after1, inner sep=0.5pt] (textafter3) {\(s_{1}'\)};
    \node[draw,shape=circle,fill=BlueViolet,inner sep=0.5pt,below = of ini1] (ini2) {};
    \node[above right = 0.1pt of ini2, inner sep=0.5pt] (textini2) {\(s_{2}\)};
    \node[right = 2cm of ini2] (after2) {};
    \node[draw,shape=circle,fill=BlueViolet,inner sep=0.5pt,below = of ini2] (ini3) {};
    \node[below left = 0.1pt of ini3, inner sep=0.5pt] (textini3) {\(s_{3}\)};
    \node[draw,shape=circle,fill=PineGreen,PineGreen,inner sep=0.5pt,right = 2cm of ini3] (after3) {};
    \node[below right = 0.1pt of after3, inner sep=0.5pt] (textafter3) {\(\color{PineGreen}s_{3}'\)};

    \draw[->,shorten >= 2pt, shorten <= 2pt,thick] (ini1) to node [above] {\(\varepsilon\)} (after1.-180);
    \draw[->,shorten >= 2pt, shorten <= 2pt,color=PineGreen] (ini3) to node [above] {\(\color{PineGreen}\varepsilon\)} (after3.-180);

    \draw[shorten >= 2pt, shorten <= 2pt,thick] (ini1) to[bend right] node [left] {\(\equiv\)} (ini3);
    \draw[shorten >= 2pt, shorten <= 2pt, densely dashed,thick] (ini2) to[bend left] node [right] {\(\sim\)} (ini3);

    \draw[shorten >= 2pt, shorten <= 2pt, PineGreen] (after1) to[bend left] node [right] {\(\color{PineGreen}\equiv\)} (after3);
    \draw[shorten >= 2pt, shorten <= 2pt, densely dashed, PineGreen] (ini2) to[bend left] node [right] {\(\color{PineGreen}\sim\)} (after3);

    \draw[dotted,thick] (-0.75,0.5) rectangle (0.5,-2.5);
    \node[left = 0.75cm of ini3, inner sep=0.5pt] (m) {\(\mathcal{M}\)};

    \draw[dotted,thick,PineGreen] (1.75,0.5) rectangle (3,-2.5);
    \node[right = 0.92cm of after3, inner sep=0.5pt] (mm) {\(\color{PineGreen}\mathcal{M}\)};
  \end{tikzpicture}
  \vspace{-0.5em}
  \caption{Silent step in strongly related states}\label{fig:recomposition-diagrams:silent-strong}
\end{subfigure}
\begin{subfigure}{0.3\textwidth}
  \centering
  \begin{tikzpicture}
    \tikzstyle{every node}=[BlueViolet]
    \tikzstyle{every path}=[BlueViolet]

    \node[draw,shape=circle,fill=BlueViolet,inner sep=0.5pt](ini1) {};
    \node[above left = 0.1pt of ini1, inner sep=0.5pt] (textini1) {\(s_{1}\)};

    \node[draw,shape=circle,fill=BlueViolet,inner sep=0.5pt,below = of ini1](ini2) {};
    \node[above right = 0.1pt of ini2, inner sep=0.5pt] (textini2) {\(s_{2}\)};

    \node[draw,shape=circle,fill=BlueViolet,inner sep=0.5pt,right = 2cm of ini2] (after2) {};
    \node[above right = 0.1pt of after2, inner sep=0.5pt] (textafter1) {\(s_{2}'\)};

    \node[draw,shape=circle,fill=BlueViolet,inner sep=0.5pt,below = of ini2] (ini3) {};
    \node[below left = 0.1pt of ini3, inner sep=0.5pt] (textini3) {\(s_{3}\)};

    \node[draw,shape=circle,fill=BlueViolet,inner sep=0.5pt,right = 2cm of ini1] (equ1) {};
    \node[above right = 0.1pt of after1, inner sep=0.5pt] (textafter3) {\(s_{1}\)};
    \draw[double,shorten >= 2pt, shorten <= 2pt] (ini1) to[bend left] node [below] {\(=\)} (equ1);

    \node[draw,shape=circle,fill=PineGreen,inner sep=0.5pt,right = 2cm of ini3] (equ3) {};
    \node[below right = 0.1pt of after3, inner sep=0.5pt] (textafter3) {\(s_{3}\)};
    \draw[double,shorten >= 2pt, shorten <= 2pt] (ini3) to[bend right] node [above] {\(=\)} (equ3);

    \draw[->,shorten >= 2pt, shorten <= 2pt,thick] (ini2) to node [above] {\(\varepsilon\)} (after2.-180);

    \draw[shorten >= 2pt, shorten <= 2pt,thick] (ini1) to[bend right] node [left] {\(\equiv\)} (ini3);
    \draw[shorten >= 2pt, shorten <= 2pt, densely dashed,thick] (ini2) to[bend left] node [right] {\(\sim\)} (ini3);

    \draw[shorten >= 2pt, shorten <= 2pt, densely dashed, PineGreen] (ini3) to[bend right] node [right] {\(\color{PineGreen}\sim\)} (after2);

    \draw[dotted,thick] (-0.75,0.5) rectangle (0.5,-2.5);
    \node[left = 0.75cm of ini3, inner sep=0.5pt] (m) {\(\mathcal{M}\)};

    \draw[dotted,thick,PineGreen] (1.75,0.5) rectangle (3,-2.5);
    \node[right = 0.92cm of after3, inner sep=0.5pt] (mm) {\(\color{PineGreen}\mathcal{M}\)};
  \end{tikzpicture}
  \vspace{-0.5em}
  \caption{Silent step in weakly related states}\label{fig:recomposition-diagrams:silent-weak}
\end{subfigure}
\begin{subfigure}{0.3\textwidth}
  \centering
  \begin{tikzpicture}
    \tikzstyle{every node}=[BlueViolet]
    \tikzstyle{every path}=[BlueViolet]

    \node[draw,shape=circle,fill=BlueViolet,inner sep=0.5pt](ini1) {};
    \node[above left = 0.1pt of ini1, inner sep=0.5pt] (textini1) {\(s_{1}\)};

    \node[draw,shape=circle,fill=BlueViolet,inner sep=0.5pt,right = 2cm of ini1] (after1) {};
    \node[above right = 0.1pt of after1, inner sep=0.5pt] (textafter1) {\(s_{1}'\)};

    \node[draw,shape=circle,fill=BlueViolet,inner sep=0.5pt,below = of ini1] (ini2) {};
    \node[above right = 0.1pt of ini2, inner sep=0.5pt] (textini2) {\(s_{2}\)};

    \node[draw,shape=circle,fill=BlueViolet,inner sep=0.5pt,right = 2cm of ini2] (after2) {};
    \node[above right = 0.1pt of after2, inner sep=0.5pt] (textafter2) {\(s_{2}'\)};

    \node[draw,shape=circle,fill=BlueViolet,inner sep=0.5pt,below = of ini2] (ini3) {};
    \node[below left = 0.1pt of ini3, inner sep=0.5pt] (textini3) {\(s_{3}\)};
    \node[draw,shape=circle,fill=PineGreen,PineGreen,inner sep=0.5pt,right = 2cm of ini3] (after3) {};
    \node[below right = 0.1pt of after3, inner sep=0.5pt] (textafter3) {\(\color{PineGreen}s_{3}'\)};

    \draw[shorten >= 2pt, shorten <= 2pt,thick] (ini1) to[bend right] node [left] {\(\equiv\)} (ini3);
    \draw[shorten >= 2pt, shorten <= 2pt, densely dashed,thick] (ini2) to[bend left] node [right] {\(\sim\)} (ini3);

    \draw[shorten >= 2pt, shorten <= 2pt, densely dashed, PineGreen] (after1) to[bend left] node [right] {\(\color{PineGreen}\sim\)} (after3);
    \draw[shorten >= 2pt, shorten <= 2pt, PineGreen] (after2) to[bend right] node [left] {\(\color{PineGreen}\equiv\)} (after3);

    \draw[->,shorten >= 2pt, shorten <= 2pt,thick] (ini1) to node [above] {\(a\)} (after1.-180);
    \draw[->,shorten >= 2pt, shorten <= 2pt,thick] (ini2) to node [above] {\(a\)} (after2.-180);
    \draw[->,shorten >= 2pt, shorten <= 2pt,color=PineGreen] (ini3) to node [above] {\(\color{PineGreen}a\)} (after3.-180);

    \draw[dotted,thick] (-0.75,0.5) rectangle (0.5,-2.5);
    \node[left = 0.75cm of ini3, inner sep=0.5pt] (m) {\(\mathcal{M}\)};

    \draw[dotted, PineGreen] (1.75,0.5) rectangle (3,-2.5);
    \node[right = 0.92cm of after3, inner sep=0.5pt] (mm) {\(\color{PineGreen}\mathcal{M}\)};
  \end{tikzpicture}
  \vspace{-0.5em}
  \caption{Non-silent step with swapping relations}\label{fig:recomposition-diagrams:event-swap}
\end{subfigure}
  \caption{Recomposition diagrams from SECOMP~\cite{ThibaultBLAAGHT24}}\label{fig:recomposition-diagrams}
\end{figure*}

To understand the general intuition behind these diagrams, recall the recomposition diagrams from \Cref{sec:new-approach}, and let us zoom in on a call and its surrounding execution.
\par
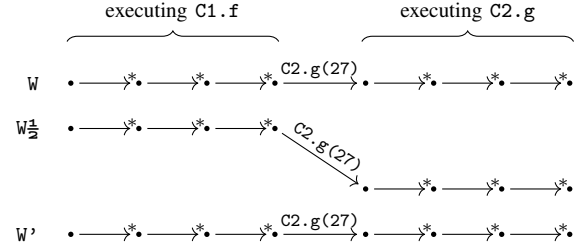
{\centering
  \begin{tikzpicture}[
    st/.style={circle,fill,inner sep=0pt,minimum size=2pt},
    tr/.style={->,shorten >=2pt,shorten <=2pt},
    star/.style={pos=1,anchor=south east,inner sep=0pt,xshift=1pt,yshift=0.5pt,font=\scriptsize}]
    \foreach \r/\name/\yf/\yg in {1/W/0/0, h/W½/-0.6/-1.4, 2/W'/-2/-2} {
      \node[left,font=\footnotesize] at (-0.3,\yf) {\code{\name}};
      \foreach \i/\x in {0/0, 1/0.9, 2/1.8, 3/2.7}
        \node[st] (\r-\i) at (\x,\yf) {};
      \foreach \i/\x in {4/3.9, 5/4.8, 6/5.7, 7/6.6}
        \node[st] (\r-\i) at (\x,\yg) {};
      \foreach \u/\v in {0/1, 1/2, 2/3, 4/5, 5/6, 6/7}
        \draw[tr] (\r-\u) -- node[star] {\(*\)} (\r-\v);
      \draw[tr] (\r-3) -- (\r-4);
    }
    \node[font=\scriptsize\ttfamily,above,inner sep=1.5pt] at ($(1-3)!0.5!(1-4)$) {C2.g(27)};
    \draw[decorate,decoration={brace,amplitude=4pt},black] (-0.05,0.55) -- (2.75,0.55)
      node[midway,above=4pt,font=\footnotesize] {executing \code{C1.f}};
    \draw[decorate,decoration={brace,amplitude=4pt},black] (3.85,0.55) -- (6.65,0.55)
      node[midway,above=4pt,font=\footnotesize] {executing \code{C2.g}};
    \path (h-3) -- node[font=\scriptsize\ttfamily,above,sloped,inner sep=1.5pt] {C2.g(27)} (h-4);
    \node[font=\scriptsize\ttfamily,above,inner sep=1.5pt] at ($(2-3)!0.5!(2-4)$) {C2.g(27)};
  \end{tikzpicture}
\par}
The executions represented by this diagram start in \code{C1.f}, and after a few silent steps, there is a call to \code{C2.g}, after which there are a few more silent steps.
The middle execution is that of \code{W½}, built by taking \code{C1} from \code{W} and \code{C2} from \code{W'}.
What happens is the following: when executing \code{W½} within \code{C1.f}, since \code{W½} took \code{C1} from \code{W}, the execution will be more similar to the one of \code{W}, and conversely, when executing within \code{C2.g}, it will be more similar to \code{W'}.
This is because we're executing the same code, and the behaviour of a compartment can only be influenced by what it can observe from other compartments, \ie what appears on the trace.

To formalize this intuition, SECOMP (taking inspiration from the ideas of El-Korashy \etal~\cite{akram-capabileptrs,secureptrs}) puts forward the idea of proving a three-way simulation using three relations: a \emph{strong simulation relation} \(\equiv\), a \emph{weak simulation relation} \(\sim\), and a \emph{mixed simulation relation} \(\mathcal{M}\).
When executing within \code{C1}, the strong relation \(\equiv\) relates the states of \code{W} and \code{W½}, and the weak relation \(\sim\) relates the states of \code{W'} and \code{W½}.
The mixed relation \(\mathcal{M}\) is a ternary relation that relates all three states, to capture invariants that are common to the three executions, such as the shape of the cross-compartment stack.

Then, SECOMP gives several diagrams to prove, such that together they imply recomposition, or in our case blame-aware recomposition.
\Cref{fig:recomposition-diagrams} depicts the core three diagrams.
In that figure, each row corresponds to one of the three executions, with the recomposed execution being at the bottom; arrows are steps (or their reflexive transitive closure); in thick dark blue are hypotheses, and in light green are conclusions to prove; and the double line represents equality.
To illustrate their use, we have overlaid these diagrams on top of the recomposition diagram:
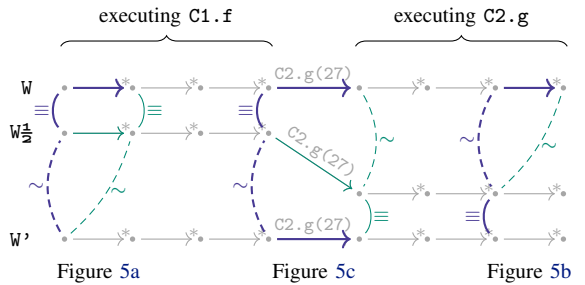
\begin{center}
  \begin{tikzpicture}[
    hyp/.style={thick,BlueViolet,shorten >=2pt,shorten <=2pt},
    concl/.style={PineGreen,shorten >=2pt,shorten <=2pt},
    rel/.style={font=\scriptsize,inner sep=1pt},
    st/.style={circle,fill,inner sep=0pt,minimum size=2pt},
    tr/.style={->,shorten >=2pt,shorten <=2pt},
    star/.style={pos=1,anchor=south east,inner sep=0pt,xshift=1pt,yshift=0.5pt,font=\scriptsize}]
    \begin{scope}[black!35]
      \foreach \r/\name/\yf/\yg in {1/W/0/0, h/W½/-0.6/-1.4, 2/W'/-2/-2} {
        \node[left,font=\footnotesize,black] at (-0.3,\yf) {\code{\name}};
        \foreach \i/\x in {0/0, 1/0.9, 2/1.8, 3/2.7}
          \node[st] (\r-\i) at (\x,\yf) {};
        \foreach \i/\x in {4/3.9, 5/4.8, 6/5.7, 7/6.6}
          \node[st] (\r-\i) at (\x,\yg) {};
        \foreach \u/\v in {0/1, 1/2, 2/3, 4/5, 5/6, 6/7}
          \draw[tr] (\r-\u) -- node[star] {\(*\)} (\r-\v);
        \draw[tr] (\r-3) -- (\r-4);
      }
      \node[font=\scriptsize\ttfamily,above,inner sep=1.5pt] at ($(1-3)!0.5!(1-4)$) {C2.g(27)};
    \draw[decorate,decoration={brace,amplitude=4pt},black] (-0.05,0.55) -- (2.75,0.55)
      node[midway,above=4pt,font=\footnotesize] {executing \code{C1.f}};
    \draw[decorate,decoration={brace,amplitude=4pt},black] (3.85,0.55) -- (6.65,0.55)
      node[midway,above=4pt,font=\footnotesize] {executing \code{C2.g}};
      \path (h-3) -- node[font=\scriptsize\ttfamily,above,sloped,inner sep=1.5pt] {C2.g(27)} (h-4);
      \node[font=\scriptsize\ttfamily,above,inner sep=1.5pt] at ($(2-3)!0.5!(2-4)$) {C2.g(27)};
    \end{scope}
    \foreach \x/\fig in {0.45/silent-strong, 3.3/event-swap, 6.15/silent-weak}
      \node[font=\footnotesize] at (\x,-2.45) {\Cref{fig:recomposition-diagrams:\fig}};
    \draw[hyp,->] (1-0) -- (1-1);
    \draw[hyp] (1-0) to[bend right=45] node[left,rel] {\(\equiv\)} (h-0);
    \draw[hyp,densely dashed] (2-0) to[bend left=30] node[left,rel] {\(\sim\)} (h-0);
    \draw[concl,->] (h-0) -- (h-1);
    \draw[concl] (1-1) to[bend left=45] node[right,rel] {\(\equiv\)} (h-1);
    \draw[concl,densely dashed] (2-0) to[bend right=15] node[right,rel] {\(\sim\)} (h-1);
    \draw[hyp,->] (1-3) -- (1-4);
    \draw[hyp,->] (2-3) -- (2-4);
    \draw[hyp] (1-3) to[bend right=45] node[left,rel] {\(\equiv\)} (h-3);
    \draw[hyp,densely dashed] (2-3) to[bend left=30] node[left,rel] {\(\sim\)} (h-3);
    \draw[concl,->] (h-3) -- (h-4);
    \draw[concl] (2-4) to[bend right=45] node[right,rel] {\(\equiv\)} (h-4);
    \draw[concl,densely dashed] (1-4) to[bend left=30] node[right,rel] {\(\sim\)} (h-4);
    \draw[hyp,->] (1-6) -- (1-7);
    \draw[hyp] (2-6) to[bend left=45] node[left,rel] {\(\equiv\)} (h-6);
    \draw[hyp,densely dashed] (1-6) to[bend right=30] node[left,rel] {\(\sim\)} (h-6);
    \draw[concl,densely dashed] (1-7) to[bend left=15] node[right,rel] {\(\sim\)} (h-6);
  \end{tikzpicture}
\end{center}

\Cref{fig:recomposition-diagrams:silent-strong} states that strongly-related executions proceed in lockstep.
In the example, we are currently executing \code{C1.f}, the states of \code{W} and \code{W½} are strongly related while the states of \code{W½} and \code{W'} are weakly related, and \code{W} takes a silent step.
We must prove that \code{W½} takes a silent step, and reestablish the relations: the new states of \code{W} and \code{W½} are strongly related, while the new state of \code{W½} is still weakly related to that of \code{W'}.

\Cref{fig:recomposition-diagrams:silent-weak} captures the fact that the weakly-related execution can step silently.
In the example, we are currently executing \code{C2.g}, the states of \code{W'} and \code{W½} are strongly related while the states of \code{W½} and \code{W} are weakly related, and \code{W} takes a silent step.
We must prove that this step preserves the weak relation, \ie that the new state of \code{W} is still related to the state of \code{W½}.
Informally, this diagram allows the weakly related execution to catch up to the other two.

Finally, \Cref{fig:recomposition-diagrams:event-swap} captures what happens during a step that produces an event: all executions must produce the same event, and the relations must be reestablished.
In the example, we are switching from \code{C1} to \code{C2}, which is exactly the boundary between the two parts that make \code{W½} up, and as a result the relations are swapped: before the call, the states of \code{W} and \code{W½} were strongly related and the states of \code{W½} and \code{W'} were weakly related, and now it must be the opposite: the new states of \code{W} and \code{W½} are now weakly related, and the states of \code{W½} and \code{W'} are now strongly related.

In particular, \Cref{fig:recomposition-diagrams:silent-weak} is the key to understanding why blame-aware recomposition works: running \(\trg{C} \linktrg \trg{P}\) can only exhibit undefined behaviour within \(\trg{C}\), and the state of \(\trg{C} \linktrg \trg{P}\) is then weakly related to the recomposed one.
This means what happens in that execution doesn't really matter: what matters is that the weak and mixed relations are preserved, and those only need to be precise enough to prove \Cref{fig:recomposition-diagrams:event-swap} when needed, which in practice lets us use relations that don't mention the internal state of the weakly related execution in a way that would be invalidated by undefined behaviour happening.

Together with a proof that the initial states satisfy the simulation relations, these diagrams can be composed one by one to imply the existence of the recomposed execution, and so they imply blame-aware recomposition.

\textbf{Defining the Simulation Relations}
We can now describe how we instantiate these three simulation relations concretely.
In our micro-policy machine, states are triples \(s = (m, \mathit{regs}, \mathtt{PC})\), and we use notations such as \(\mathit{colour}\) or \(\mathit{code}\) to inspect fields of the tags.

How do we define the strong relation? First, the strong and weak relations should only describe one side of the split between partial program and context, so we parameterize them with \(\mathbb{I} = \mathbb{P} | \mathbb{C}\), and we use \(\equiv_{\mathbb{P}}\) and \(\sim_{\mathbb{C}}\) when the execution is in \(\mathbb{P}\), and \(\equiv_{\mathbb{C}}\) and \(\sim_{\mathbb{P}}\) when the execution is in \(\mathbb{C}\), as depicted below.
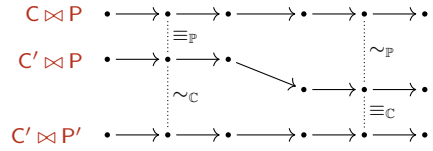
\begin{center}
  \begin{tikzpicture}[
    st/.style={circle,fill,inner sep=0pt,minimum size=2pt},
    tr/.style={->,shorten >=2pt,shorten <=2pt},
    rl/.style={densely dotted,shorten >=1.5pt,shorten <=1.5pt},
    rel/.style={font=\scriptsize,inner sep=1.5pt}]
    \foreach \r/\name/\yb/\ya in {1/{\trg{C} \linktrg \trg{P}}/0/0, h/{\trg{C'} \linktrg \trg{P}}/-0.6/-1.0, 2/{\trg{C'} \linktrg \trg{P'}}/-1.6/-1.6} {
      \node[left,font=\footnotesize] at (-0.2,\yb) {\(\name\)};
      \foreach \i/\x in {0/0, 1/0.8, 2/1.6}
        \node[st] (\r-\i) at (\x,\yb) {};
      \foreach \i/\x in {3/2.6, 4/3.4, 5/4.2}
        \node[st] (\r-\i) at (\x,\ya) {};
      \foreach \u/\v in {0/1, 1/2, 2/3, 3/4, 4/5}
        \draw[tr] (\r-\u) -- (\r-\v);
    }
    \draw[rl] (1-1) -- node[rel,right] {\(\equiv_{\mathbb{P}}\)} (h-1);
    \draw[rl] (h-1) -- node[rel,right] {\(\sim_{\mathbb{C}}\)} (2-1);
    \draw[rl] (1-4) -- node[rel,right] {\(\sim_{\mathbb{P}}\)} (h-4);
    \draw[rl] (h-4) -- node[rel,right] {\(\equiv_{\mathbb{C}}\)} (2-4);
  \end{tikzpicture}
\end{center}

We define \(\equiv_{\mathbb{I}}\) as follows:
\anndecl{rel:1}{1}{}\anndecl{rel:2}{2}{}\anndecl{rel:3}{3}{}%
\anndecl{rel:4}{4}{}%
\begin{mathpar}
  \mprset{vskip=1pt}
  \inferrule{
    \annbox{1}{1}{\(\mathit{colour}(\mathtt{PC}) = \mathit{colour}(\mathtt{PC}') = c \quad c \in \mathbb{I}\)} \\
    \annbox{2}{2}{\(\mathtt{PC}' = \mathtt{PC} + \mathit{offset}(c)\)} \\
    \annbox{3}{3}{\(\mathit{regs} \simeq_{\mathbb{I}} \mathit{regs}'\)} \\
    \annbox{4}{4}{\(m \simeq_{\mathbb{I}} m'\)}
  }{
    (m, \mathit{regs}, \mathtt{PC}) \equiv_{\mathbb{I}} (m', \mathit{regs}', \mathtt{PC}')
  }
\end{mathpar}
During the simulation, strongly related states are always in the same compartment, one that belongs to \(\mathbb{I}\), as captured by condition~\cref{rel:1}.
More than this, strongly related states are always executing the same instruction, since they follow the same code.
However, the condition \(\mathtt{PC}' = \mathtt{PC}\) would be incorrect.
Indeed, the code layout might have changed, and so actually there is a fixed offset for each compartment \(c\), \(\mathit{offset}(c)\), between the original programs and the recomposed one.
Hence condition~\cref{rel:2}.

Then, to know that the same instructions are going to lead to the same results (otherwise the two strongly related executions would not be in lockstep), we need to know that the content of the registers and of the memory is the same: this is captured by \cref{rel:3} and \cref{rel:4}.
The relations \(\simeq_{\mathbb{I}}\) on registers and memories relate the content of these registers and memories: for instance, we define \(\simeq_{\mathbb{I}}\) on memories to be:
\begin{align*}
  m \simeq_{\mathbb{I}} m' \triangleq (\forall w\,d.~ &\mathit{colour}(d)\in\mathbb{I} \wedge \neg \mathit{code}(d) \implies\\
  &m(w) = d \implies m'(w) = d)\,\wedge \\
  (\forall w\,d'.~ &\mathit{colour}(d')\in\mathbb{I} \wedge \neg \mathit{code}(d') \implies\\
  &m'(w) = d' \implies m(w) = d') \wedge
  \dots
\end{align*}
We have elided some well-formedness predicates on the memories and registers that state that the return addresses stored in registers and in memory point to code belonging to the same compartment.
Any piece of data in \(m\) belonging to a compartment that is in \(\mathbb{I}\) can be found at the same address in \(m'\): this captures the fact that strongly related compartments have done the exact same sequence of memory operations, and is enabled by the allocation mechanism hiding the data layout of each compartment.

The weak relation is defined in a similar manner:
\begin{mathpar}
  \mprset{vskip=1pt}
  \inferrule{
    \annbox{1}{1}{\(\mathit{colour}(\mathtt{PC}) = \mathit{colour}(\mathtt{PC}') = c \quad c \notin \mathbb{I}\)} \\
    \annbox{4}{4}{\(m \simeq_{\mathbb{I}} m'\)}
  }{
    (m, \mathit{regs}, \mathtt{PC}) \sim_{\mathbb{I}} (m', \mathit{regs}', \mathtt{PC}')
  }
\end{mathpar}
Note that we keep \cref{rel:1} (except that the states that are related are the ones not executing) and \cref{rel:4}, but we drop the requirements on the PC and on the registers.
Indeed, in the weakly related states, the execution doesn't proceed in lockstep (recall the diagrams from the previous section), and so the internal details don't need to be synchronized.
We keep conditions on the memory, but note that the weak relation mentions state outside the one that is executing.

The key takeaway for both these relations is that they only quantify over the parts of the programs that are actually kept in the recomposed execution:
when the execution is inside \(\mathbb{P}\), then \(\equiv_{\mathbb{P}}\) is used for the strong relation and relates the parts of the state that belong to \(\trg{P}\), and \(\sim_{\mathbb{C}}\) is used for the weak relation and relates the parts of the state that belong to \(\trg{C'}\).

\anndecl[1]{rel:a}{}{a}\anndecl[2]{rel:b}{}{b}%
\begin{mathpar}
  \mprset{vskip=1pt}
  \inferrule{
    \annbox{1}{a}{\(\mathit{tag}(\mathtt{PC}) = \mathit{tag}(\mathtt{PC}') = \mathit{tag}(\mathtt{PC}'') = (\mathtt{Level\ n}, c)\)} \\
    \annbox{2}{b}{\(\mathbb{S}_n(s, m, m', m'')\)} \\
    \dots
  }{
    \mathcal{M}_{s}\;(m, \mathit{regs}, \mathtt{PC})\;(m', \mathit{regs}', \mathtt{PC}')\;(m'', \mathit{regs}'', \mathtt{PC}'')
  }
\end{mathpar}
Finally, we also have a mixed relation, which is parameterized by a shadow stack \(s\), where stack frames quadruples containing three return addresses and one colour.
The shadow stack keeps track of the cross-compartment return addresses, ensuring they are in sync.
This is captured by the two conditions \cref{rel:a}, which enforce that the stack is at the same level in all three executions, and \cref{rel:b}, which is defined by the two following rules.
\anndecl[3]{rel:c}{}{c}%
\anndecl[4]{rel:d}{}{d}%
\begin{mathpar}
  \mprset{vskip=1pt}
  \inferrule{ }{
    \mathbb{S}_0(\cdot, m, m', m'')
  }
  \and
  \inferrule{
    \annbox{3}{c}{\(\begin{array}{@{}c@{}}
      c \in \mathbb{P} \implies v_c = v \qquad
      c \in \mathbb{C} \implies v_c = v' \\
      v'' = v_c + \mathit{offset}(c)
    \end{array}\)} \\\\
    \annbox{4}{d}{\(\mathit{code}_c(m, v) \quad \mathit{code}_c(m', v') \quad \mathit{code}_c(m'', v'')\)} \\\\
    \mathbb{S}_n(s, m, m', m'')
  }{
    \mathbb{S}_{n+1}((v, v', v'', c) :: s, m, m', m'')
  }
\end{mathpar}
The rule for the empty stack is clear.
For adding a stack frame on the stack, the rule is slightly more involved.
Conditions \cref{rel:d} capture the fact that the return addresses point to code inside \(c\).
Then, following \cref{rel:c}, how the return addresses are related depends on which compartment they target.
If they target a compartment in \(\mathbb{P}\), then the return address \(v''\) in the recomposed execution must be related to \(v\), so that they point to the same instruction (with an offset because the code layout has changed).
Conversely, if they target a compartment in \(\mathbb{C}\), then \(v''\) must be related to \(v'\), not \(v\).

Reasoning about these three relations actually requires reasoning about the compartmentalization mechanism that is implemented: for instance, one of the key invariants that we need to establish to prove blame-aware recomposition is the uniqueness of the linear return capabilities.
Another one is that we need to reason about the memory layout and consider various edge cases, such as adding conditions on the size of the code (it should not be too large or it would overflow into the data region).
This supports our argument that blame-aware recomposition allows us to reason about a lower-level backend compared to the previous proof technique.

\subsection{Formalization}
We have implemented the allocator changes to the micro-policy backend of Abate \etal.
We have used our technique to prove \(\rscdcmd\) between our overlay semantics and the low-level semantics of our updated backend.
We have also proved \(\rscdcmd\) between the source language of Abate \etal and our updated micro-policy backend, under a few additional assumptions, mainly that the source executions do not go beyond the limits of the machine (our source uses unbounded integers, whereas our backend uses concrete 32-bit integers).
Everything is done within the Rocq proof assistant, and the full development can be found as supplementary materials.

\section{Related Work}\label{sec:related}

We build on top of the work of Abate \etal~\cite{AbateABEFHLPST18,Abate_2018_arxiv}.
Researchers have since developed various other extensions of this work focused on adding new features or proving stronger results. For instance, SECOMP~\cite{ThibaultBLAAGHT24} focuses on scaling to a realistic source language, by extending CompCert~\cite{Leroy06} with a compartmentalization mechanism and proving it secure. Thibault \etal~\cite{ThibaultLH25} propose an extension of the back-translation that lets them prove a stronger security result that requires code confidentiality. El-Korashy \etal~\cite{secureptrs} add the possibility for components to share memory by exchanging pointers. Both of these extensions are seemingly incompatible with our work: the first one requires recomposition to hold at the lowest level considered, while with the second one, exchanging pointers irremediably leaks some of the components' internal behaviour if they can be freely compared.

Building and verifying low-level backends for compilers is a long-standing research problem.
The state-of-the-art verified C compiler, CompCert~\cite{Leroy06,Leroy09b}, targets relatively high-level backends that use a block-based memory model.
As a result, it is impossible to model many behaviours such as buffer overflows or stack overflows.
Various extensions of CompCert make progress towards more realism: CompCertMC~\cite{WangWS19} extends CompCert with lower-level backends that have flat memory spaces. Its successor CompCertELF~\cite{WangXWS20} adds a verified assembler that produces ELF object files, concatenating the code together with data.
However, all these works focus on correct compilation only, and in particular they do not support linking with arbitrary code.

Other verification projects have focused on low-level compartmentalization, and in particular the verification of SFI schemes, either at the assembly level~\cite{MorrisettTTTG12,ZhaoLSR11} or at higher levels~\cite{BessonBDJW19,KrollSA14}.
The techniques SFI uses to isolate different compartments' memory regions from each other are similar to the changes we make to the backend's memory allocation scheme.
In our setting, however, because we have access to a monitor and a tagging mechanism, we do not need to introduce all the complexity that is inherent to SFI-based solutions.

Overlay semantics were introduced by Skorstengaard \etal~\cite{SkorstengaardDB19} in the context of secure compilation, to formally characterize the security of their calling convention: their overlay semantics has built-in local-state encapsulation and well-bracketed control flow.
Later, Georges \etal~\cite{GeorgesTB22} also added temporal elements to their overlay semantics, allowing them to characterize some temporal aspects of the stack safety.
In both works, the authors prove that their concrete implementation of the calling convention is fully abstract \wrt their overlay semantics.
In~\cite{KolosickNJWLGJS22}, Kolosick \etal use overlay semantics to characterize the well-behavedness of Wasm programs. They use this semantics to justify the correctness of the removal of some protections in the Wasm binaries.
The main difference between all these works and ours is that while they use overlay semantics as a characterization of the expected behaviour of their programs, we use them merely as a proof technique. More plainly, our \(\rscdcmd\) theorem is a direct statement about the kind of properties that are preserved by the compiler, while one needs to understand and trust their overlay semantics in order to understand their theorems.

\section{Conclusion and Future Work}\label{sec:future}

We have proposed a proof technique for verifying low-level secure compiler backends.
Our proof technique is a variant of a previous proof technique from the literature, and relies on blame-aware recomposition, a technical lemma that allows recomposing programs together.
The idea behind it is to characterize some problematic behaviours for the
previous proof technique as undefined behaviours \wrt an overlay semantics, and
to exploit the fact that some
of the programs that are manipulated are actually
compiled, and thus free of such undefined behaviours.
We have reported on the application of that proof technique to an existing backend using micro-policies (essentially tag-based monitors) as a protection mechanism to implement compartmentalization.

There are still several questions that we believe are worth studying, and it would be nice to investigate how far one can take our proof technique, in different directions.

For instance, how does the proof technique scale to larger and more realistic projects? SECOMP~\cite{ThibaultBLAAGHT24} showed that it is possible to prove \(\rscdcmd\) for a realistic compartmentalizing compiler; is it still the case if one wishes to go lower-level?
Another question is whether we could adapt our proofs to the SFI backend of Abate \etal.
The difficulty here is that the invariants of SFI are much more complex than ours, since we benefit from having fine-grained tags.

Another direction would be to study what criteria can be proved with this proof technique. The one we consider in this work, \(\rscdcmd\), does not depend on the confidentiality of the data or the code being run.
Yet, the backend must still enforce some confidentiality of data for the proof technique to be applicable; can one then also obtain a stronger security criterion, in the manner of~\cite{ThibaultLH25}?

A final direction is what features can be added to the languages while preserving the ability to apply the proof technique. Memory sharing~\cite{akram-capabileptrs,secureptrs} is one such feature: while it would be desirable to share memory directly between compartments, directly sharing pointers might re-introduce leaks concerning the allocation pattern of compartments.

Finally, as we mentioned, recomposition and blame-aware recomposition still rely on code confidentiality, or at least the ability to build an overlay semantics with code confidentiality. It would be more satisfying to have a direct proof technique that doesn't require this.

\newcommand{\aiusage}{%
GenAI tools were not used to write this paper, except for fixing grammar and spelling, for improving readability and tightening some existing paragraphs, for LaTeX programming, and for generating some of the figures from pen-and-paper sketches.
Most of the proof work, including the generic proof technique and the main recomposition diagrams, was initially completed by the authors without GenAI tools.
Claude Opus was later used to prove auxiliary technical lemmas (results about machine word arithmetic and bitwise operations) as well as to prove symmetric variants of some of the diagrams and clean up proofs.
}
  \section*{Acknowledgments}
  This work was in part funded by the European Research Council under ERC Starting Grant SECOMP (715753),
  and by the Deutsche Forschungsgemeinschaft (DFG, German Research Foundation)
  under Germany's Excellence Strategy -- EXC 2092 CASA -- 390781972.
  This research is supported by armasuisse Science and Technology.

  \paragraph{AI Usage}
  \aiusage

\bibliographystyle{abbrvnaturl}
\bibliography{paper,mp,safe}
\end{document}